%% file: sample-sigconf.tex
\documentclass[sigconf,authorversion]{acmart}

\usepackage[ruled,linesnumbered, vlined]{algorithm2e}
\usepackage{longtable,supertabular,booktabs}
\usepackage{threeparttablex,caption}
\usepackage{multirow}
\usepackage{bbding}
\usepackage{wasysym}
\usepackage{iitem}
\usepackage{latexsym}
\usepackage{url}
\usepackage{pifont}
\usepackage[shortlabels]{enumitem}
\usepackage{subcaption}
\usepackage{amsthm}
\usepackage[normalem]{ulem}
\usepackage{makecell}
\usepackage[export]{adjustbox}
\usepackage{balance}
\usepackage{xspace}

\newcommand{\ie}{i.e.,\xspace}
\newcommand{\eg}{e.g.,\xspace}

\let\oldnl\nl% Store \nl in \oldnl
\newcommand{\nonl}{\renewcommand{\nl}{\let\nl\oldnl}}% Remove line number for one line
\newcommand{\sys}{\textsl{\mbox{EnclaveTree}}\xspace}
\newcommand{\tabincell}[2]{\begin{tabular}{@{}#1@{}}#2\end{tabular}} % split rows in table

\AtBeginDocument{%
  \providecommand\BibTeX{{%
    \normalfont B\kern-0.5em{\scshape i\kern-0.25em b}\kern-0.8em\TeX}}}

\copyrightyear{2022}
\acmYear{2022}
\setcopyright{acmcopyright}\acmConference[ASIA CCS '22]{Proceedings of the 2022 ACM Asia Conference on Computer and Communications Security}{May 30-June 3, 2022}{Nagasaki, Japan}
\acmBooktitle{Proceedings of the 2022 ACM Asia Conference on Computer and Communications Security (ASIA CCS '22), May 30-June 3, 2022, Nagasaki, Japan}
\acmPrice{15.00}
\acmDOI{10.1145/3488932.3517391}
\acmISBN{978-1-4503-9140-5/22/05}

\begin{document}
\fancyhead{}
%%
%% The "title" command has an optional parameter,
%% allowing the author to define a "short title" to be used in page headers.
\title{EnclaveTree: Privacy-preserving Data Stream Training and Inference Using TEE}

%%
%% The "author" command and its associated commands are used to define
%% the authors and their affiliations.
%% Of note is the shared affiliation of the first two authors, and the
%% "authornote" and "authornotemark" commands
%% used to denote shared contribution to the research.

\author{Qifan Wang}
\affiliation{%
  \institution{The University of Auckland}
  \city{Auckland}
  \country{New Zealand}}
\email{qwan301@aucklanduni.ac.nz}

\author{Shujie Cui}
\affiliation{%
  \institution{Monash University}
  \city{Melbourne}
  \country{Australia}}
\email{shujie.cui@monash.edu}

\author{Lei Zhou}
\affiliation{%
  \institution{Southern University of Science and Technology}
  \city{Shenzhen}
  \country{China}}
\email{zhoul6@sustech.edu.cn}

\author{Ocean Wu}
\affiliation{%
  \institution{The University of Auckland}
  \city{Auckland}
  \country{New Zealand}}
\email{hwu344@aucklanduni.ac.nz}

\author{Yonghua Zhu}
\affiliation{%
  \institution{The University of Auckland}
  \city{Auckland}
  \country{New Zealand}}
\email{yzhu970@aucklanduni.ac.nz}

\author{Giovanni Russello}
\affiliation{%
  \institution{The University of Auckland}
  \city{Auckland}
  \country{New Zealand}}
\email{g.russello@auckland.ac.nz}

%%
%% By default, the full list of authors will be used in the page
%% headers. Often, this list is too long, and will overlap
%% other information printed in the page headers. This command allows
%% the author to define a more concise list
%% of authors' names for this purpose.
\renewcommand{\shortauthors}{Wang, et al.}

\begin{abstract}
The classification service over a stream of data is becoming an important offering for cloud providers, but users may encounter obstacles in providing sensitive data due to privacy concerns. 
While Trusted Execution Environments (TEEs) are promising solutions for protecting private data, they remain vulnerable to \textit{side-channel attacks} induced by data-dependent access patterns. 
We propose a Privacy-preserving Data Stream Training and Inference scheme, called \textit{\sys}, that provides confidentiality for user's data and the target models against a compromised cloud service provider. 
We design a matrix-based training and inference procedure to train the Hoeffding Tree (HT) model and perform inference with the trained model inside the trusted area of TEEs, which provably prevent the exploitation of access-pattern-based attacks. 
%\sys can be easily extended to Random Forest (RF) model by training each tree of the RF with the HT training and \revision{predicting} unlabelled data by converting a set of trees into our matrix representation and performing the HT inference. 
The performance evaluation shows that \sys is practical for processing the data streams with small or medium number of features. When there are less than 63 binary features, \sys is up to ${\thicksim}10{\times}$ and ${\thicksim}9{\times}$ faster than na\"ive oblivious solution on training and inference, respectively. 

\end{abstract}

%%
%% The code below is generated by the tool at http://dl.acm.org/ccs.cfm.
%% Please copy and paste the code instead of the example below.
%%


%%
%% Keywords. The author(s) should pick words that accurately describe
%% the work being presented. Separate the keywords with commas.

%%
%% This command processes the author and affiliation and title
%% information and builds the first part of the formatted document.
\maketitle

\input{sections/introduction}
\input{sections/preliminary}

\input{sections/model}

\input{sections/tree}
\input{sections/ppht}

\input{sections/implementation}

%\input{sections/RF_perf}
\input{sections/related}
\input{sections/conclusion}

\begin{acks}
Russello would like to acknowledge the MBIE-funded programme STRATUS (UOWX1503) for its support and inspiration for this research.
\end{acks}

\bibliographystyle{ACM-Reference-Format}
\bibliography{ref}
%\balance
%%
%% If your work has an appendix, this is the place to put it.
\appendix
\input{sections/appendix}

\end{document}

%% file: sections/introduction.tex
\section{Introduction} 
Machine learning (ML) applications such as remote healthcare and activity recognition, have attracted a lot of attention as a major breakthrough in the practice of ML. 
These specific applications of ML are characterized by \emph{data streams}: data is generated by various devices and usually arrive in a timely manner. 
For either training the ML model or inferring (\ie predicting or evaluating) an unlabelled instance by the model, the data stream should be processed efficiently on-the-fly as data might arrive rapidly. 
The Hoeffding Tree (HT) model \cite{domingos2000mining}, a variation of the decision tree model, has become the standard for processing data streams. 

%Nevertheless, due to limited resources, it is infeasible to build the model or provide real-time inference on the data owner. 
To process data streams efficiently, a promising solution is to outsource the HT training and inference to cloud platforms \cite{nguyen2015survey,wang2019privstream}. However, this poses a severe threat to data privacy and model confidentiality. 
For privacy-sensitive applications, such as in the health-care domain, all the data samples, the model, the inference output, and any intermediate data generated during the model training and inference should be protected from the Cloud Service Provider (CSP). 
In particular, when training a HT, the main operation is to classify each newly arriving data sample with the \textit{current tree} and to count the frequency of different feature values. %among the data samples classified into each leaf node. 
The access path over the tree and the statistical information generated when training the model, should be protected as they can be leveraged by an adversary to construct a near-equivalent HT \cite{tramer2016stealing}. 
\begin{table*}[!ht]
\footnotesize
\centering
 \caption{Comparison of decision tree training and inference protocols}
 \label{table:S1}
% \begin{ThreePartTable}
 \renewcommand\tabcolsep{3.5pt} 
    \begin{tabular}{ccccccccccc}
        \toprule
        \multirow{2}{*}{\textbf{Scheme}} & \multicolumn{3}{c}{\textbf{Support}} & \multicolumn{2}{c}{\textbf{Communication}} & \multirow{2}{*}{\textbf{Complexity on Client}} & \multicolumn{4}{c}{\textbf{Privacy}} \\
        
        \cmidrule(lr){2-4}\cmidrule(lr){5-6}\cmidrule(lr){8-11}
        & \textbf{DS} & \textbf{Training} & \textbf{Inference} & \textbf{Rounds} & \textbf{Bandwidth} & & \textbf{Data} & \textbf{IR} & \textbf{Model} & \textbf{AP} \\
        \midrule
        
        Du et al. \cite{du2002building}, Vaidya et al. \cite{vaidya2005privacy}, Samet et al. \cite{samet2008privacy} & \multirow{3}{*}{\XSolidBrush} & \multirow{3}{*}{$\checkmark$} & \multirow{3}{*}{\XSolidBrush} & \multirow{3}{*}{$\Omega(t_d)$} & \multirow{3}{*}{$\Omega({t_m}logn+n)$} & \multirow{3}{*}{$\Omega(({t_m}+1)n)$} & \CIRCLE & \CIRCLE & \Circle & \Circle \\
        
        \cline{1-1} \cline{8-11}
        Xiao et al. \cite{xiao2005privacy}, Emekci et al. \cite{emekcci2007privacy} & & & & & & & \CIRCLE & \LEFTcircle & \CIRCLE & \CIRCLE \\
        
        \cline{1-1} \cline{8-11}
        Hoogh et al. \cite{de2014practical}, Lindell et al. \cite{lindell2000privacy} & & & & & & & \CIRCLE & \CIRCLE & \CIRCLE & \CIRCLE \\
        \hline
        
        \tabincell{c}{Bost et al. \cite{bost2015machine}, Wu et al. \cite{wu2016privately}, Tai et al. \cite{tai2017privacy}, \\ Kiss et al. \cite{kiss2019sok} } & \multirow{2}{*}{\XSolidBrush} & \multirow{2}{*}{\XSolidBrush} & \multirow{2}{*}{$\checkmark$} & $c{\ge}2$ & \multirow{2}{*}{$\Omega(t_m+n)$} & \multirow{2}{*}{$\Omega(t_m+n)$} & \multirow{2}{*}{\CIRCLE} & \multirow{2}{*}{\CIRCLE} & \multirow{2}{*}{\LEFTcircle} & \multirow{2}{*}{\CIRCLE} \\
        
        \cline{1-1} \cline{5-5}
        Cock et al. \cite{de2017efficient} & &  &  & $t_b+3$ &  &  &  &  &  &  \\
        \hline
        
        Akavia et al. \cite{akavia2019privacy} & \XSolidBrush & $\checkmark$ & $\checkmark$ & $t_d$ & $O(t_m+n)$ & $O(t_m+n)$ & \CIRCLE & \CIRCLE & \CIRCLE & \CIRCLE\\
        \hline
        
        Liu et al. \cite{liu2020towards} & \XSolidBrush & $\checkmark$ & $\checkmark$ & $1$ & $O(n)$ & $O(n)$ & \CIRCLE & \CIRCLE & \LEFTcircle & \Circle \\
        \hline 
        \textbf{\sys} & $\checkmark$ & $\checkmark$ & $\checkmark$ & $1$ & $O(n)$ & $O(n)$ & \CIRCLE & \CIRCLE & \CIRCLE & \CIRCLE \\
        \bottomrule
    \end{tabular}
       
       \textbf{DS} denotes data stream.
       Privacy of data, intermediate results, model and access patterns are denoted by \textbf{Data}, \textbf{IR}, \textbf{Model} and \textbf{AP}, respectively.
       $\Omega(\cdot)$ and $O(\cdot)$ denote the computation complexity of each party in the distributed setting and client, respectively. 
       \CIRCLE, \LEFTcircle, and \Circle ~denote the target is protected, part of the parameters of the target are leaked, and fails to protect the target, respectively.
       $t_d$, $t_m$, $t_b$, $c$ and $n$ represents the tree's depth, the number of nodes, the binary representation length of data samples, constants and the number of data samples, respectively.
 %\end{ThreePartTable}
\end{table*}
%Data Owner (\textit{Definition} in Section \ref{sec: system_model}) is denoted as Client; 
%

\textit{Privacy-preserving data mining (PPDM)} aims to protect the privacy of outsourced ML tasks by employing cryptographic primitives, such as Secure Multi-Party Computation (SMC) \cite{lindell2000privacy,du2002building,vaidya2005privacy,xiao2005privacy,de2017efficient,zheng2019towards} or Homomorphic Encryption (HE) \cite{akavia2019privacy,liu2020towards,bost2015machine,wu2016privately}. 
Nevertheless, most of the existing PPDM approaches cannot be adopted to process data streams, because they:  
\ding{182} cannot process complicated functions such as logarithm and exponential operations in an efficient way, which are fundamental to the HT model training; 
\ding{183} impose too heavy computation and communication overheads on the clients;. %, making the approach impractical for resource-constraint devices; 
\ding{184} leak statistical information and tree structures. 

Table~\ref{table:S1} summarizes the related work in this area. First of all, note that most of the existing approaches focus on generic decision trees and none of them can securely process data streams (column DS in Table~\ref{table:S1}). The approaches given in \cite{lindell2000privacy,du2002building,vaidya2005privacy,xiao2005privacy,emekcci2007privacy,samet2008privacy,de2014practical} are impractical for data streams because they require multiple rounds of interactions between client and server. While the approaches proposed in \cite{bost2015machine,wu2016privately,de2017efficient,zheng2019towards} leak information about the model, such as the structure and the number of nodes of the tree. 
To the best of our knowledge, \cite{wang2019privstream, xu2008privacy} are the only approaches that focus on data streams and can provide some level of protection for the data, the target model and the inference results.
The reason these works are not included in the table is because they do not focus on decision trees. Moreover, the main idea of these approaches is to randomly perturb the data distribution with noise. This approach is usually efficient but at the cost of accuracy loss due to a large amount of perturbations.
Furthermore, since only part of information is perturbed, the attacker can still compromise the user's privacy by retrieving the features through inference attacks \cite{aggarwal2005k}. 

\noindent\textbf{Our goals.}
In this work, we aim to design an outsourced approach to train and infer data streams with HT model in a secure and efficient manner. 
Specifically, our approach should not only protect all the data samples and the model from the CSP but also any intermediate data generated during the training and inference, such as the frequency of different feature values and the access pattern. 
%Finally, our approach should be practical so that client devices with limited resources can be used to generate encrypted data streams. 

\noindent\textbf{Challenges.}
To achieve the goals, we propose a privacy-preserving data stream classification scheme called \textit{\sys}. The basic idea of \sys is to employ the Intel Software Guard Extension (SGX) \cite{costan2016intel} to process privacy-sensitive operations on the CSP. 
Intel SGX is an extension of the x86 instruction set architecture that allows a user process to create trusted execution environments called \textit{enclaves} on the CSP. %which protect security-critical operations from other privileged software, including the OS kernel and hypervisor.  
Recent work \cite{schuster2015vc3,ohrimenko2016oblivious,law2020secure,poddar2020visor} has demonstrated that SGX-based PPDM is orders of magnitude faster than cryptography-based approaches. 
Moreover, within an enclave, one can process any kind of operations securely and efficiently, including logarithm and exponentiation. 
However, using Intel SGX is non-trivial because it suffers from side-channel attacks, which enable an adversary to obtain the access pattern over HT and then infer secrets \cite{ohrimenko2016oblivious,rane2015raccoon,law2020secure}, \eg the tree structure. For instance, with controlled-channel attack \cite{xu2015controlled}, the adversary can learn which pages are accessed when classifying a data sample. By injecting enough malicious data samples, the adversary could recover the tree structure.  %based on the access pattern at page level. 
Thus, the challenge of using Intel SGX is to protect the enclave access pattern.

The traditional method for data classification with tree models is to traverse the tree from the root to a leaf node by comparing a node with the corresponding feature value level by level. %However, an adversary can easily obtain the accessed path with side-channel attacks, which is sensitive information that can be leveraged to recover data samples or the tree structure. 
To protect the access pattern, a na\"ive solution can be implemented by accessing the node in each level obliviously. For instance, using a solution as proposed by \cite{law2020secure} we could store the nodes at each level of the tree as an array and then obliviously access the target node in the array to update the statistical information. 
However, this approach is costly. 

\noindent\textbf{Our Contributions.}
The contributions of this paper are threefold. 

First of all, we are the first to propose a secure and efficient scheme to process data streams for decision tree models in outsourced environments. As shown in Table~\ref{table:S1}, compared with existing PPDM schemes for the decision tree model, \sys not only achieves better communication and computation overhead, but also achieves better security guarantees.
To the best of our knowledge, \sys is the first scheme that can efficiently and securely process data streams with protection for data samples, the model, statistical information, and tree access pattern. Moreover, \sys imposes a very light overhead on client devices, where only standard encryption operations are required for outsourcing data samples for processing and decrypting the results after inference.

Our second contribution is a novel approach for tree classification based on matrix multiplications. 
Inspired by the approach in \cite{poddar2020visor}, \sys performs the HT training by periodically reading a batch of data samples, converting them into a matrix $\mathcal{M}_d$, transforms the current model into a matrix $\mathcal{M}_q$, and updates the frequency of different feature values by computing $\mathcal{M}_d \times \mathcal{M}_q$. 
The main advantage of our approach is that inherently it does not leak any access pattern and is more efficient than traversing the tree using oblivious operations. 
Similarly, \sys also  classifies unlabelled instances with a matrix multiplication. 

%\notegio{Given that we have space why don't we list the contributions?}

%the IBM HELib \cite{halevi2020design} and
We implemented the prototype of \sys with OpenEnclave \cite{microsoftOE} and evaluated its performance. 
The results show that, \sys takes about 6.73, 29.4, and 134 seconds to process $5{\times}10^4$ data samples with 15, 31, 63 features respectively, which is ${10.4\times}$, $4.2{\times}$, $1.1{\times}$ faster than the na\"ive oblivious solution. 
As for HT inference, \sys takes 1.89, 2.80, and 4.72 milliseconds for inferring 100 unlabelled instances with a tree of depth 9, and outperforms the na\"ive oblivious solution by ${9.2\times}$, $7.2{\times}$, $6.5{\times}$ when there are 15, 31, 63 features, respectively. 

% Our contributions can be summarised as below:
%  \begin{itemize}
%      \item
     
%      \item We implemented and evaluated the performance of \sys. The result shows that 
% \end{itemize}

% \noindent\textbf{Roadmap.}
% The rest of the paper is organized as follows: Section \ref{sec: background} introduces the background; 
% Section \ref{sec: system_model} presents the system model and threat model; Section \ref{sec: data} details the representation of data and model in \sys; Section \ref{sec:details} describes the details of \sys; Section \ref{sec: implementation} provides the experiment design and reports the performance evaluations; Section \ref{sec: related_work} surveys the related works and Section \ref{sec: conclusion} concludes our work.

%% file: sections/preliminary.tex
\section{Background}
\label{sec: background}
% Labelled: British English; Labeled: American English.
% unlabelled: British spelling; unlabeled: American spelling
In this section, we provide background information on Intel SGX, side-channel attacks, and the oblivious primitives we use in the rest of this paper. 

\subsection{Intel SGX}

A Trusted Execution Environment (TEE), such as the Intel Software Guard Extensions (Intel SGX) \cite{costan2016intel}, protects sensitive data and code from privileged attackers who may control all the software, including the operating system and hypervisor. In Intel SGX-enabled machines, the CPU protects the confidentiality and integrity of code and data by storing them in an isolated memory region, called enclave. %Specifically, the enclave memory is located in the EPC, which is protected by the Memory Encryption Engine (MEE) using cryptographic primitives \cite{costan2016intel}. The maximum available EPC memory of SGX v1 is 128MB. The data is available in plaintext inside the CPU while it is encrypted when written to EPC. %Random Access Memory (RAM). 
Intel SGX also supports remote attestation of an initialized enclave. It enables a remote party to verify an enclave identity and the integrity of the code and data inside the enclave.

%\sout{Fundamental to the remote attestation is a cryptographic signature containing the hash of an enclave layout and memory contents. The remote party could refuse to load her private data into an enclave when the corresponding hash does not match the expected one.}
\subsection{Side-channel Attacks on Intel SGX}
One issue of Intel SGX is that it still shares many resources with untrusted programs, \eg CPU cache and branch prediction units, and relies on the underlying OS for resource management. 
As a result, Intel SGX is susceptible to side-channel attacks. 
In recent years, various side channels have been extensively exploited to infer secrets from enclaves, such as L1 cache~\cite{moghimi2017cachezoom,GotzfriedESM17}, page tables~\cite{xu2015controlled,BulckWKPS17}, branch predictor~\cite{EvtyushkinRAP18,Bluethunder,SGKKP17}, and the transient execution mechanism~\cite{cacheout,Spectre,ridl}. 
They infer secrets by mainly exploiting the data-dependent enclave access pattern at different granularity. 
For instance, with cache-timing attacks, the adversary can learn the enclave access pattern at cache line granularity. 

Existing countermeasures are either hardware-based~\cite{OrenbachBS20,Strackxabs} or software-based~\cite{Varys,Hyperspace,OBFUSCURO}. Hardware-based solutions, such as cache partitioning~\cite{COLORIS} and enclave self-paging~\cite{OrenbachBS20}, are efficient yet they require hardware modifications, which take a long period to be applied and cannot be retrofitted to existing hardware. 
In contrast, software-based solutions are more flexible. However, they generally leverage expensive normalisation or randomisation techniques, making them impractical. 
For instance, OBFUSCURO~\cite{OBFUSCURO} leverages ORAM operations to perform
secure code execution and data access, which adds about $51\times$ overhead to enclaves. 
It is desirable to protect sensitive data and operations from side-channel attacks with techniques that are specific to the enclave.  
%\notewang{revision, in common 4, they mention that we cannot defend against some side-channel attacks such as cache attacks.}

\subsection{Oblivious Primitives}
\label{subsubsec:obl_primitive}
A library of general-purpose oblivious primitives, operating solely on registers whose contents are restricted to the code outside the enclave, has been introduced in previous work \cite{ohrimenko2016oblivious,rane2015raccoon,law2020secure} and experimentally demonstrated that it is several orders of magnitude faster than previous ORAM-based approaches. 
In this work, we will use the following oblivious primitives:

\begin{itemize}
    \item \textbf{Oblivious comparison.} $\mathtt{ogreater}$ and $\mathtt{oequal}$, are used to compare variables and implemented with x86 instruction $\mathtt{cmp}$.
    \item \textbf{Oblivious selection.}  $\mathtt{oselect}$, allows to conditionally select an element.
    \item \textbf{Oblivious assignment.} $\mathtt{oassign}$, allows to conditionally assign variables. It specifically uses $\mathtt{CMOVZ}$ for equality comparisons and subsequently combine it with $\mathtt{oassign}$ to assign a value to the destination register.
    %\item \textbf{Oblivious select.} $\mathtt{oselect}$ selects either $x$ or $y$ according to the condition result.
    % \item \textbf{Oblivious sort.} Similar to \cite{ohrimenko2016oblivious,law2020secure}, $\mathtt{osort}$ obliviously sorts the array with $n$ values based on compare-and-swap functions through a bitonic sorting network \cite{batcher1968sorting}. \notewang{remove.}
    \item \textbf{Oblivious array access.} $\mathtt{oaccess}$ scans the array at cache-line granularity and obliviously load one element based on $\mathtt{oassign}$, and then is optimized with AVX2 vector instructions \cite{corparation2016intel}.
\end{itemize}

%% file: sections/model.tex
\section{Overview of Our Approach}
\label{sec: system_model}
In this section, we will describe the system  model and design overview. We will conclude the section with discussing the threat model. 

\subsection{System Model}

\begin{figure}[!t]
\centering
\includegraphics[width=0.7\linewidth]{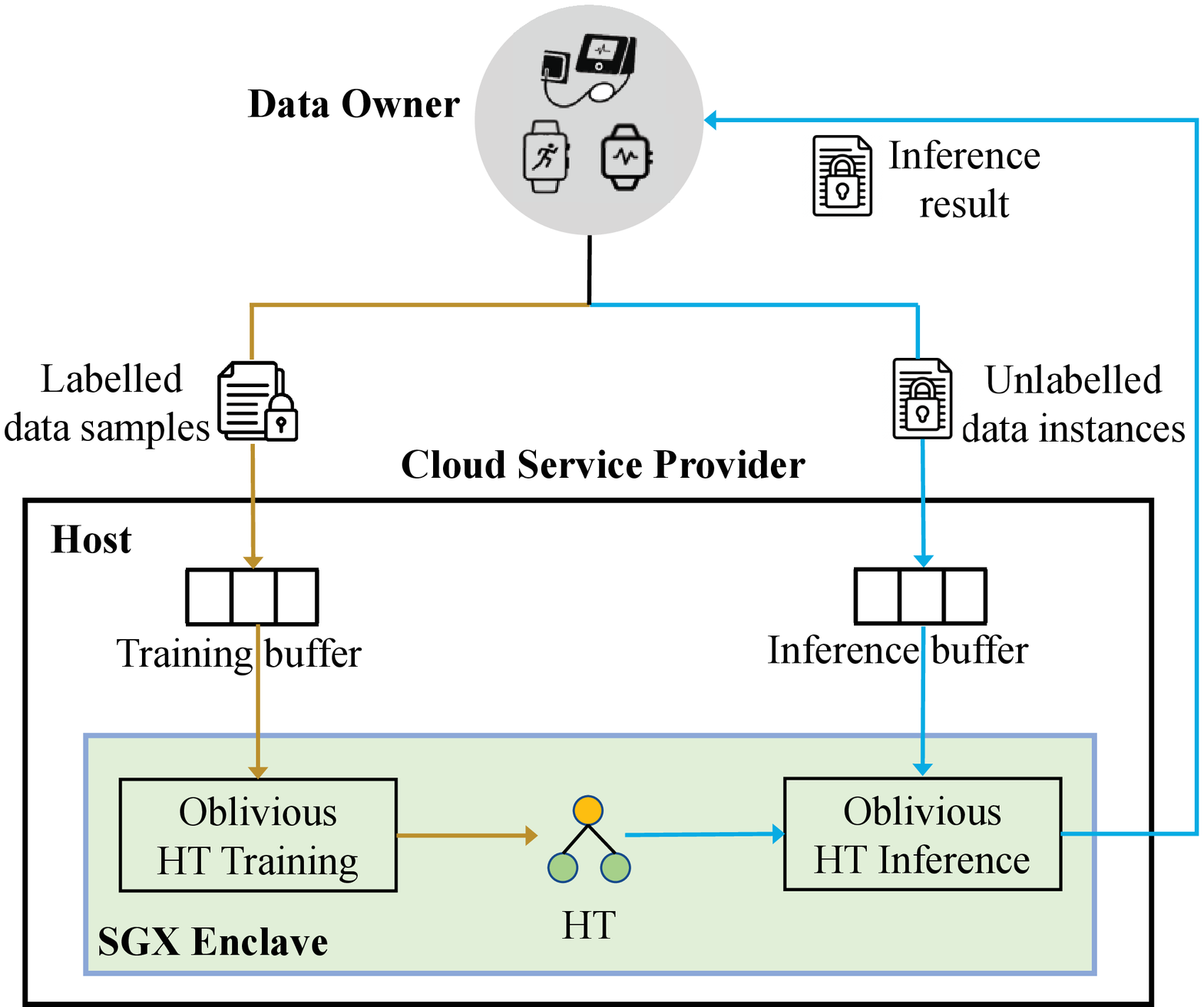}
\caption{The architecture of \sys.} 
\label{fig:arch}
\end{figure}

As shown in Fig.~\ref{fig:arch}, \sys consists of 2 entities: the Data Owner (DO) and the Cloud Service Provider (CSP). 
%\sout{Note that we employ the symmetric-key variant SHE in our work and assume that data is end-to-end encrypted at the $DO$.}

\begin{itemize}
    \item The DO continuously receives data from devices, encrypts them, and outsources them to the CSP. The data samples could be labelled samples or unlabelled instances. 
    In particular, labelled samples are used to train the model, while the unlabelled instances will be inferred with a label value by the model. The inference results are sent from the CSP to the DO.
    
    \item The CSP considered in \sys should have Intel SGX support, (\eg Microsoft Azure \cite{microsoft2021} and Alibaba Cloud \cite{ali2020white}). 
    The CSP consists of a trusted and an untrusted component. 
    The trusted component is represented by the \emph{SGX enclave} (as shown in Fig.~\ref{fig:arch}). This is where the models are trained and where the inference is performed. 
    The untrusted component is any computational resources in the \emph{CSP Host} that is outside the \emph{SGX Enclave}. 
    With the assistance of an enclave, the CSP trains the model with the data samples outsourced from the DO and classify them with the model.
    
\end{itemize}

\subsection{Design Overview}
\label{subsec: design overview}
The architecture of \sys is shown in Fig.~\ref{fig:arch}. 
It consists of two sub-components within the enclave: \textbf{Oblivious Training} and \textbf{Oblivious Inference}, and two buffers: \textbf{Training Buffer} and \textbf{Inference Buffer} outside the enclave. 
The two buffers outside the enclave receive encrypted labelled and unlabelled data from the DO for training and inference, respectively.  
Each sub-component reads data from the corresponding buffer periodically for subsequent processing. 
Oblivious Training outputs the HT model which will be the input of the Oblivious Inference. 
For unlabelled data instances, Oblivious Inference returns the predication results to the DO. 
%\revision{The two components can be deployed in the same enclave or two separate enclaves.}

To protect the access pattern for both training and inference, the two tasks are converted into matrix multiplications. 
Moreover, we use oblivious primitives, such as $\mathtt{ogreater}$, $\mathtt{oassign}$ and $\mathtt{oaccess}$, to process the remaining operations in order to hide the enclave memory access pattern.

%\subsection{Threat Model and Security Guarantees}
\subsection{Threat Model}
\label{threat model}
We assume the \textit{DO} and the SGX enclave are fully trusted. 
The CSP host is untrusted and attempts to infer secrets, such as the tree structure and statistical information, by observing and analysing memory access pattern of the enclave. 
Moreover, the CSP can eavesdrop on the communication between the DO and the enclave. 
%However, we assume the CSP does not modify or replay the data sent from the DO to the enclave.
Note that rollback attacks \cite{parno2011memoir}, denial-of-service attacks \cite{gruss2018another} and other attacks based on physical information, such as electromagnetic, power consumption and acoustic are out of our scope.

The security analysis of \sys is given in Appendix~\ref{sec:proof}.

%% file: sections/tree.tex
%
% \begin{figure*}[htp]
% \centering
% \centering
% \includegraphics[width=\textwidth]{graph/TreeConstruct.eps}
% \caption{The HTs before and after one around of training.}
% \label{fig: enclave storage}
% \end{figure*}

% \begin{figure}[!t]
% \centering
% \includegraphics[width=\columnwidth]{graph/Tree.eps}
% \caption{A HT example and its extension after one round of training. \revision{The orange nodes are internal nodes which have been assigned with features. The green nodes are leaves. Each leaf has a label for inference, and it stores statistical information for the features that have not been assigned to the path for training. Each branch is assigned with a feature value.}}
% %\notegio{IN THE TEXT DESCRIBING THIS FIGURE YOU NEED TO SAY SOMETHING ABOUT THE TWO DIFFERENT COLOURS OF THE NODES (ORANGE AND GREEN)}
% \label{fig:tree example}
% \end{figure}

% \begin{figure}[!ht]
%   \centering
%       \includegraphics[width=\linewidth]{graph/TreeConstruct1.eps}
%     % \includegraphics[width=\linewidth]{graph/TreeConstruct1_test.eps}
% \caption{The HT representation in \sys. The columns and the arrays in blue are dummy ones.}
% \label{fig: tree matrix}
% \end{figure}

\section{Data and Model Representation}
\label{sec: data}
\begin{table}[!t]
	\centering
	\footnotesize
	\caption{Notations}%\notegio{MOVE THIS CLOSER TO SECTION 4.2}
	\label{tab:notation}
	%\begin{ThreePartTable}
	\begin{tabular}{c l}
		\toprule
		\textbf{Notation} & \textbf{Description} \\ \hline
		$D$ & A data sample \\	\hline
		$d$ & Number of features \\	\hline
		$S$ & A sequence of $d$ features, $S=(s_1, s_2, {\cdots}, s_d)$ \\	\hline
		$m_i$ & Number of values of feature $s_i$, where $i \in [1,d]$ \\	\hline
		$V_{s_i}$ & Values of $s_i$, $V_{s_i}=(v_{i,1}, v_{i,2}, {\cdots}, v_{i, m_i})$ \\	\hline
		$M$ & Length of the bit-representation of $D$ \\	\hline
		$\overline{G}(\cdot)$ & Heuristic measure, \ie Information Gain (IG) \\ \hline
		%$\epsilon$ & \textcolor{red}{Hoeffding bound} \\	\hline
		$P_{real}/P_{dummy}$ & Number of real/dummy paths in the tree \\	\hline
		$\mathcal{M}_d/\mathcal{M}_i$ & Labelled/Unlabelled data matrix \\	\hline
% 		$\mathcal{M}_i$ & Unlabelled data instance matrix \\	\hline
		$\mathcal{M}_t$ & Matrix representation of the model \\ \hline
% 		a $(M-m_d) \times P$ tree matrix, where $P=P_{real}+P_{dummy}$
		
		$\mathcal{M}_t[;p]$ & The $p$-th column of $\mathcal{M}_t$ \\ \hline 
		$u_p$ & Number of unassigned features for $\mathcal{M}_t[;p]$ \\ \hline
		$\tau_p$ & Number of assigned feature values for $\mathcal{M}_t[;p]$ \\ \hline
		$\mathcal{M}_q^p$ & Query matrix of $\mathcal{M}_t[;p]$ \\	\hline
		%\notewang{$\mathcal{M}_r/\mathcal{M}_r'$} & \makecell[l]{The result matrix for $\mathcal{M}_d \times \mathcal{M}_q$ and \\ $\mathcal{M}_i \times \mathcal{M}_t$, respectively} \\ \hline
		$L$ & Number of all possible $(value, label)$ \\	\hline
		$L'$ & {\makecell[l]{Number of all possible $(value, label)$ \\ of unassigned features in a leaf}} \\	\hline
		$c_{(value, label)}$ & Frequency of each $(value, label)$ \\	%\hline
	%	$Leaf[p][l]$ & The $l$-th $c_{(value, label)}$ of the $p$-th leaf \\ 	\hline
%		\notewang{$\sigma({\cdot})$} & a mapping $\sigma$: $k$ in $\mathcal{M}_r[;k] {\to} l$ in $Leaf[p][l]$ \\
		\bottomrule
	\end{tabular}
	%\end{ThreePartTable}
\end{table}
%
% \begin{figure}[!t]
% \centering
% \centering
% \includegraphics[width=\columnwidth]{graph/TreeRep.eps}
% \caption{An example of a tree training and its representation within the enclave.}
% \label{fig: tree example}
% \end{figure}
%\notegio{This figure contains too many details - I am not sure if this is the right way to present the "simple" tree to the reader. Maybe we need just a simpler figure with just the tree to describe it to the reader first}. \notewang{Todo} 

% \notewang{We define the mapping between $\mathcal{M}_r[;k]$ and $Leaf[p][l]$ as follows: given a set $F$ and $fidx{\in}[1,r_p]$, $F_{fidx}$ is the index of $r_p$-th unassigned features in $S$. 
% Therefore, we can define $l=(\sum_{i=1}^{i=F_{fidx}} m_i*m_d)+k-(\sum_{j=0}^{j=r_p-1} m_{F_{j}}*m_d)$, where by default $m_{F_0}$ is 0.}

In this section, we provide some details on how a Hoeffding Tree (HT) is originally built and used for inference. Then, we will describe how the data and the HT are represented in our approach. To make things more concrete, we will use a simple running example throughout this paper. The example consists of building a HT to decide whether it is suitable to play tennis based on a weather dataset \cite{emekcci2007privacy}. The example tree consists of 4 features and each feature has 2 or 3 possible values listed as follows: $Outlook~(Sunny, Overcast, Rain)$, $Windy~(True, False)$, $Humidity~(High, Normal)$, and $Temperature$ $(Hot, Cool)$. Each internal node of the tree is assigned with a feature, and its possible values determine the branches of the node. The leaf nodes represent the label that has 2 values: either $Yes$ or $No$. Fig.~\ref{fig:tree example} shows how a HT is built for this example. 
In the rest of this paper we will use the notation shown in Table~\ref{tab:notation}.

\subsection{Hoeffding Tree}
\label{subsec:HT construction}
\begin{figure}[ht!]
  \centering
    \begin{subfigure}{0.48\columnwidth}
      \centering   
      %[width=\linewidth]
      % [width=4.1cm,height=2.8cm]
      \includegraphics[width=4.1cm,height=2.8cm]{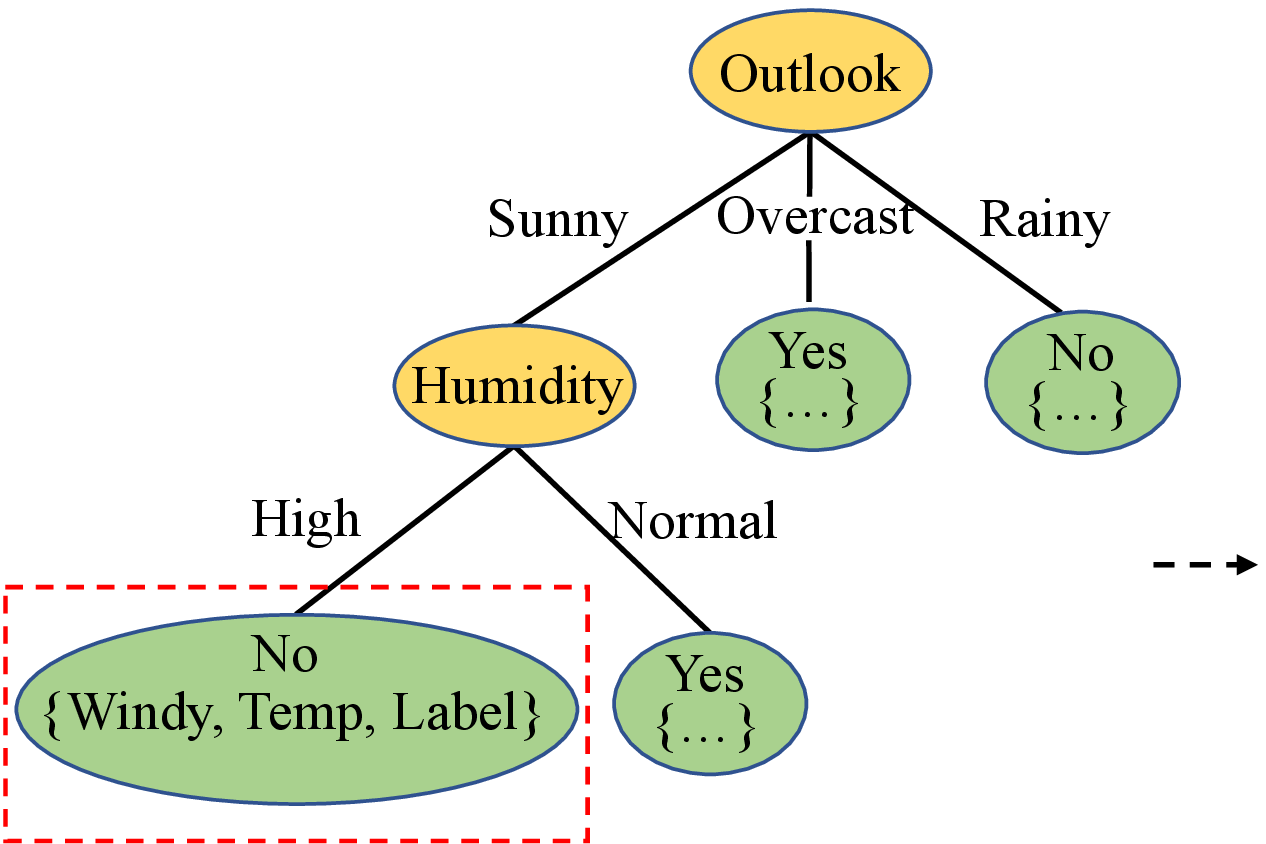}
        \caption{Before one round training}
        \label{fig:sub1}
    \end{subfigure}   %      \hfill  % 
    \begin{subfigure}{0.48\columnwidth}
      \centering   
      \includegraphics[width=4.1cm,height=2.8cm]{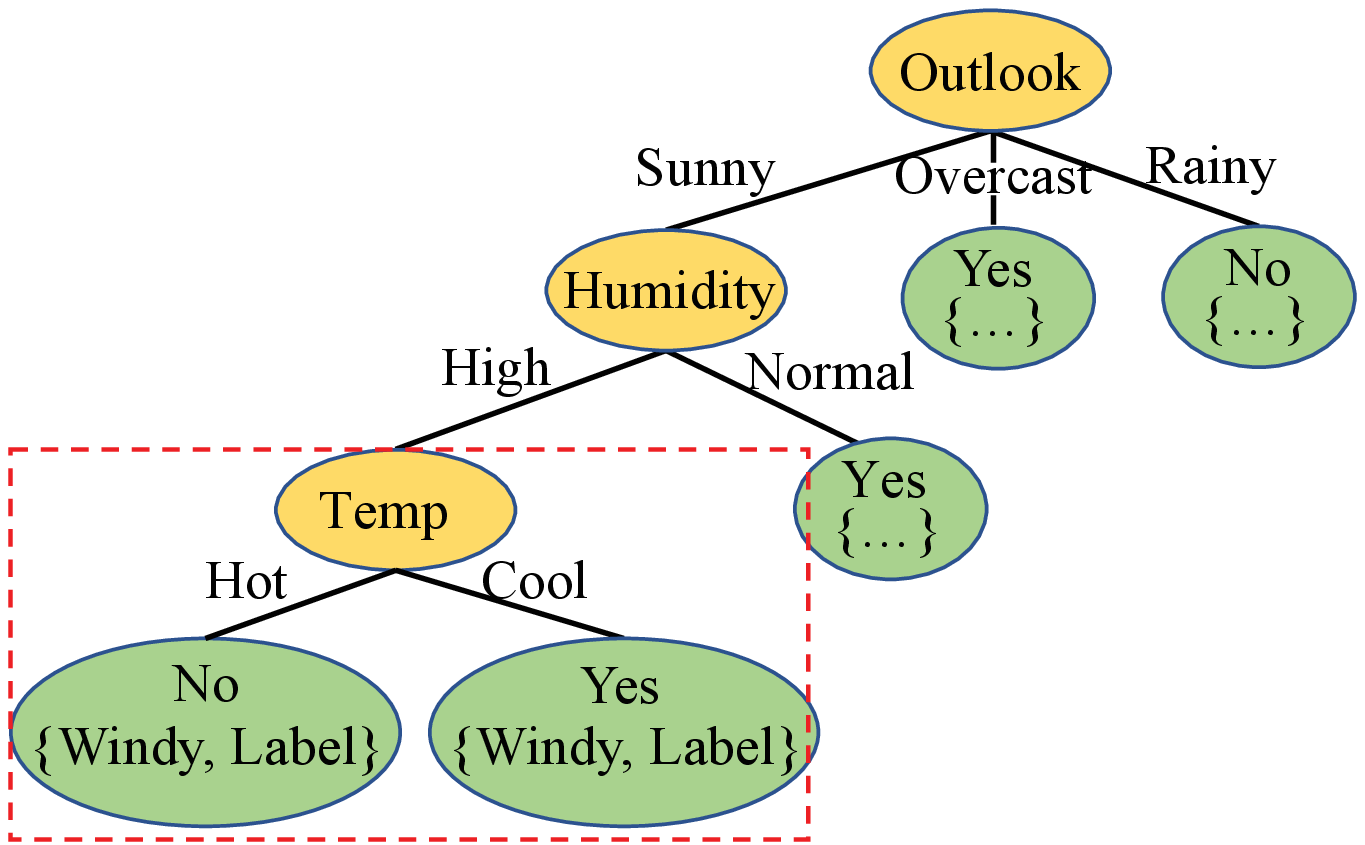}
        \caption{After one round training}
        \label{fig:sub2}
    \end{subfigure}
\caption{
A HT example and its extension after one round of training. The orange nodes are internal nodes which have been assigned with features. The green nodes are leaves. Each leaf has a label value for inference, and it stores statistical information for the features that have not been assigned to the path for training. Each branch is assigned with a feature value.}
\label{fig:tree example}
\end{figure}
A decision tree consists of internal nodes (including the root) and leaves, where each internal node is associated with a test on a feature, each branch represents the outcome of the test, and each leaf represents a class label which is the decision taken after testing all the features on the corresponding path.

Building a decision tree is a process of assigning features to its internal nodes. 
The distribution of features among the nodes determines the structure of the tree which affects the inference accuracy of the model.
Thus, the key operation of tree building is to find the \textit{best feature} for each leaf, so that we can achieve a high inference accuracy. 

The HT \cite{domingos2000mining} algorithm builds a decision tree incrementally in a top-down manner, by continually converting leaves into internal nodes with the data stream. 
Converting a leaf to an internal node requires assigning a feature to the node. 
HT training uses traditional Information Gain (IG) in ID3 \cite{quinlan1986induction} and \textit{Hoeffding Bound} \cite{domingos2000mining} to evaluate which feature is the best to be assigned to a leaf for data streams. 
%C4.5
Specifically, when data samples are classified into a leaf, we compute the IG of the features that have not been assigned to any internal node on the path to the leaf, and check if the difference between the top two IGs is greater than the Hoeffding Bound. 
If yes, the feature with the highest IG will be used to covert the leaf to an internal node. 
Otherwise, we just update the statistical information stored for the leaf. 
For instance, in the tree shown in Fig.~\ref{fig:sub1}, when there are samples classified into the left-most leaf, only the IGs of $Windy$ and $Temp$ will be computed as $Outlook$ and $Humidity$ have already been assigned to internal nodes on the left-most path. 
%\notegio{Note that this notation is not defined in Table 2} 
We use $\overline{G}(feature)$ to represent the feature's IG. 
Assume $\overline{G}(Temp) > \overline{G}(Windy)$, 
$Temp$ will be chosen as the best feature for the leaf when $\overline{G}(Temp)-\overline{G}(Windy)>{\epsilon}$, where $\epsilon$ is the Hoeffding Bound. 
The tree in Fig.~\ref{fig:sub2} shows how the leaf is converted into an internal node with feature $Temp$. The new internal node generates 2 branches and 2 new leaves as $Temp$ has 2 possible values: $Hot$ and $Cool$. The two new leaves just need to compute the IG of the last unassigned feature $Windy$ for upcoming data samples.

Formally, the Hoeffding Bound is defined as ${\epsilon}=\sqrt{\frac{\log^2 c ~*~ ln(1/{\delta})}{2n}}$ \cite{domingos2000mining}, where $n$ is the number of samples classified into the leaf, $c$ is the number of total label values, and $1-\delta$ represents the probability of choosing the correct feature for the leaf node. 
Both $c$ and $\delta$ are constant.

For HT training, each leaf of the current tree keeps receiving labelled data samples, and the samples might contain different values for each feature and label. 
The IG of a feature is derived from the frequencies of its possible $(value, label)$ pairs. 
For instance, for the left-most leaf of the tree in Fig.~\ref{fig:sub1}, to compute $\overline{G}(Temp)$ we need to count how many samples have been classified into the left-most leaf. These samples might contain the following pairs: \textit{(Hot, Yes), (Hot, No), (Cold, Yes), (Cold, No)}. 
Similarly, to compute $\overline{G}(Windy)$, we need to count how many samples have got the following pairs: \textit{(True, No), (True, Yes), (False, No)}, and \textit{(False, Yes)}. 
Each leaf records the frequencies of the pairs for unassigned features and updates them when receiving new samples.
% \notewang{modify in 4.3}
%Once the best feature is determined, its possible values determine the outgoing branches of the node and generate 3 new leaves. New data samples will be classified into the 3 new leaves accordingly for the next round of training. 

Computing IG values is expensive due to the complex logarithm and exponentiation operations. Therefore, feature' IGs of each leaf are computed when the leaf receives every $n_{min}$ samples, where $n_{min}$ is a pre-defined parameter.

Overall, we can summarize the main operations of building a HT model with the following steps: 
\begin{enumerate}
    \item classifying new arrivals into leaves with current HT model; and performing steps 2 and 3 for each leaf that gets new data samples; 
    
    \item updating the frequency of each $(value, label)$ pair for unassigned features; 
    
    \item checking if the leaf has received $n_{min}$ data samples, and performing steps 4-6 if true; 
    
    \item computing the IG value for each unassigned feature; 

    \item checking if the top two highest IG values satisfy the Hoeffding Bound; 
   
    \item if true, converting the leaf node into an internal one using the feature with the highest IG value.
    %\notewang{otherwise waiting for next $n_{min}$ samples.}
    %otherwise setting the number of samples received by the leaf to 0.
\end{enumerate}

%for those who have received no less than $n_{min}$ data samples
 
%In HT algorithm, the first 3 steps are performed after receiving every data sample.   
%By setting a big value for $n_{min}$, the training performance will be improved, yet the best moment to convert the leaf node could be missed, which influences the accuracy of the model. 
%Based on the frequency of sample arrivals, it is possible to experimentally define a proper $n_{min}$ to balance the performance and the accuracy of the model. 

Inference operations start from the root of the tree. 
An unlabelled instance is tested with the feature at each internal node and then moved down the tree along the edge corresponding to the instance's value for that feature. 
When a leaf node is reached on the path, the label associated with it is assigned to the instance. 
%\notegio{Maybe could be beneficial to have here a concrete example?}\notescui{This part is easy to understand. I think the example is unnecessary.}

\subsection{Data Representation}
\label{subsec:sample}
\begin{figure}[ht!]
  \centering
    \begin{subfigure}{1\columnwidth}
      \centering   
      %[width=\linewidth]
      \includegraphics[width=\linewidth]{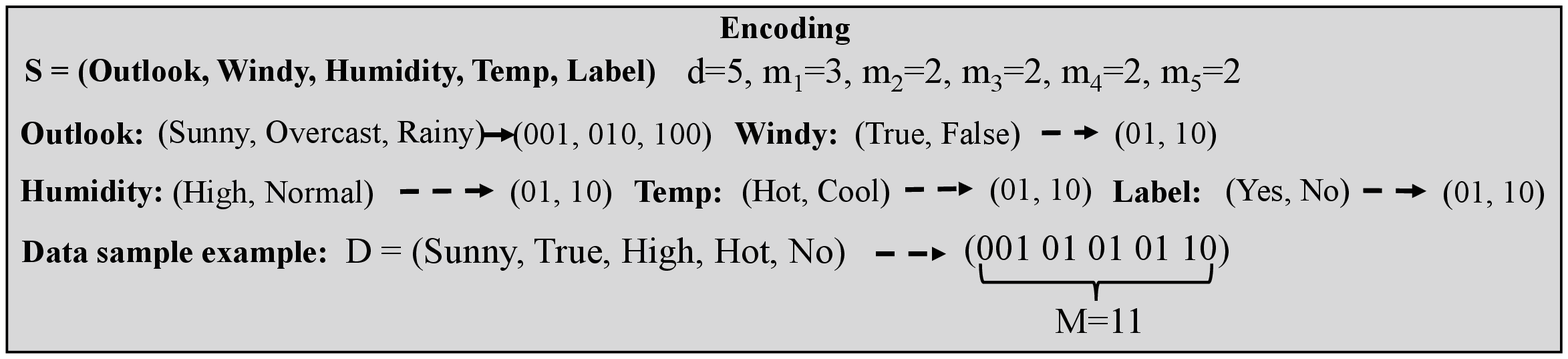}
        \caption{}
        \label{fig:treematrix_sub1}
    \end{subfigure}   %      \hfill  % 
    \begin{subfigure}{1\columnwidth}
      \centering   
      \includegraphics[width=\linewidth]{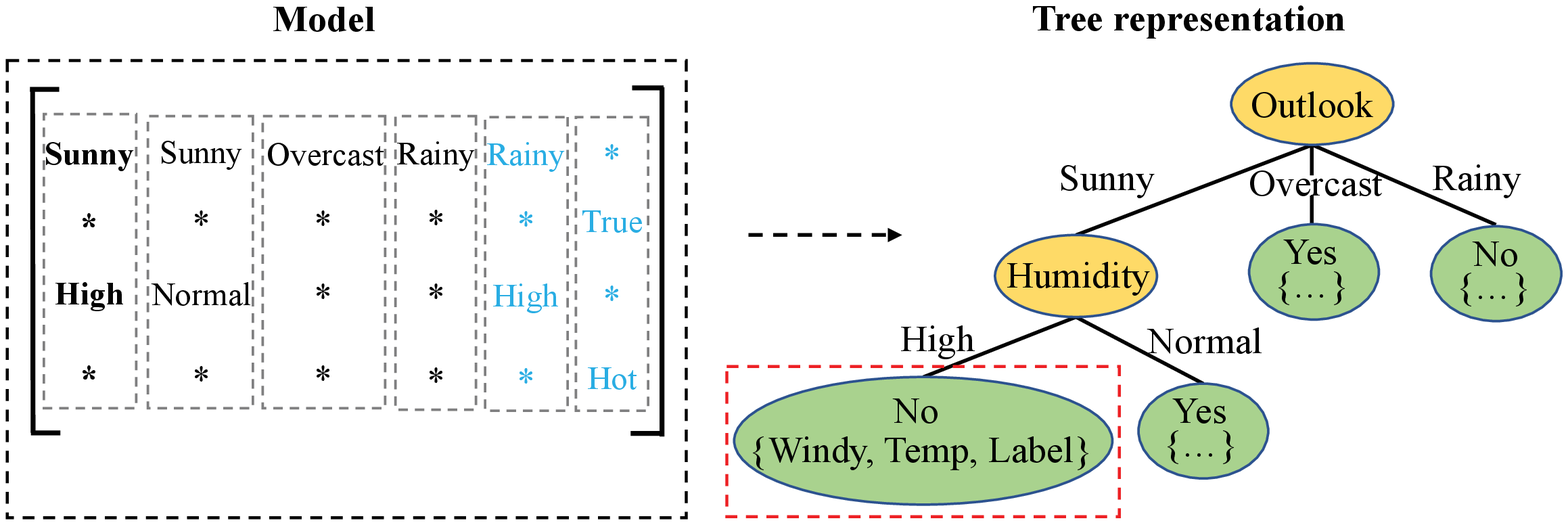}
        \caption{}
        \label{fig:treematrix_sub2}
    \end{subfigure}
    \begin{subfigure}{1\columnwidth}
    %   \centering  
      \includegraphics[width=0.65\linewidth, left]{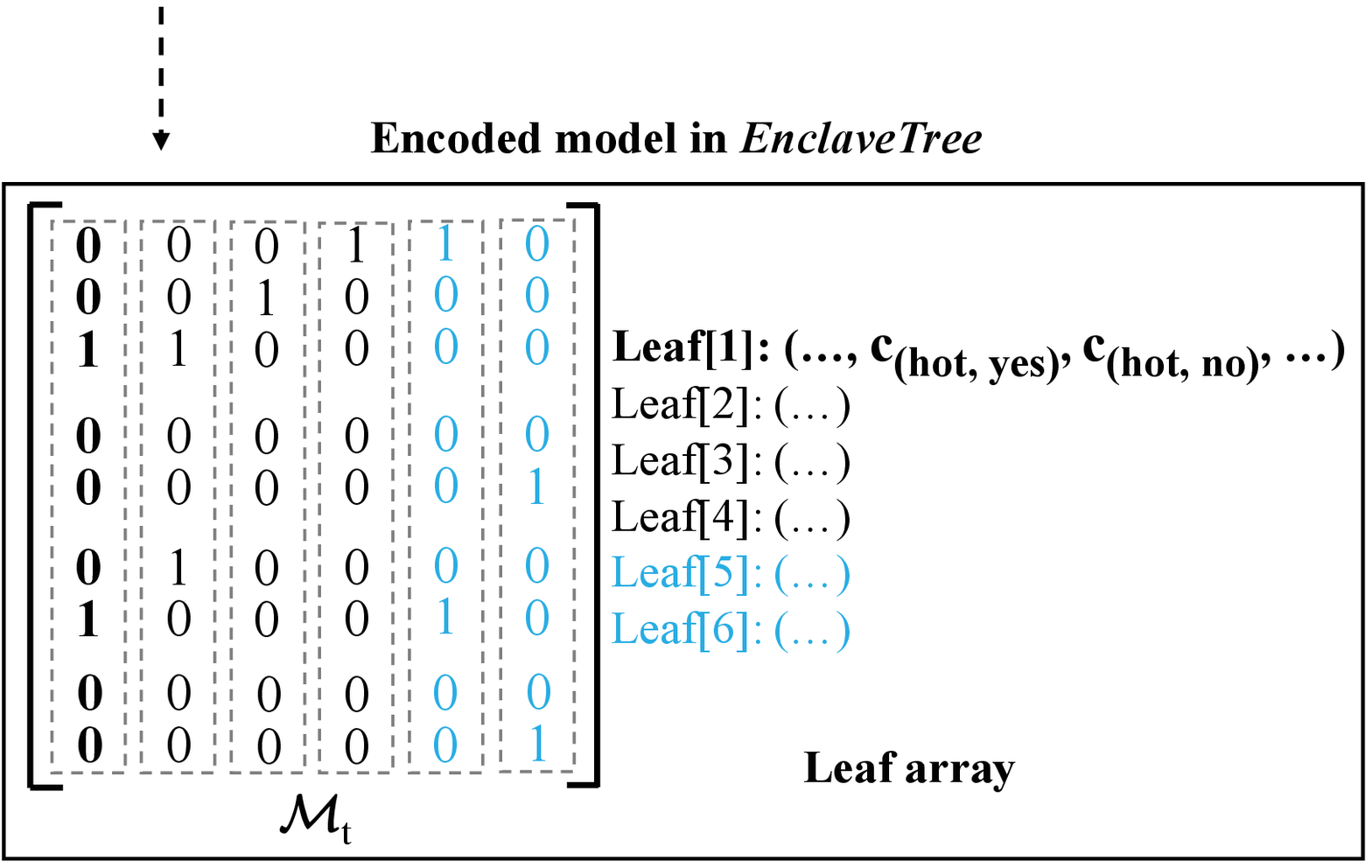}
        \caption{}
        \label{fig:treematrix_sub3}
    \end{subfigure}
\caption{The HT representation in \sys. The columns and the arrays in blue are dummy ones.}
\label{fig:tree matrix}
\end{figure}
\sloppy We assume that $S=(s_1, s_2, {\cdots}, s_d)$ represents a sequence of $d$ features, with each feature $s_i$ having $m_i$ possible values: $V_{s_i}=(v_{i,1}, v_{i,2}, {\cdots}, v_{i, m_i})$, where $1 \leq i \leq d$. 
We use the one-hot encoding technique \cite{harris2010digital} to encode each value $v_{i, j}$ into a bit string, where $1 \leq j \leq m_i$. 
More precisely, a $m_i$-bit string is used to represent a value $v_{i, j}$, where the $j$-th bit of the string is $1$ and all the other bits are $0$. 
For labelled data samples, the last feature $s_d$ is the label, and we will use the same bit representation for possible label values. 
Therefore, a data sample $D$ is represented as a bit string with $M=\sum_{i=1}^{i=d} m_i$ bits. 

Fig.~\ref{fig:treematrix_sub1} shows a concrete example for the encoding of 5 features. In the example, $d=5$ and $S=(Outlook$, $Windy$, $Humidity$, $Temp$, $Label)$. 
The first feature $s_1=Outlook$ has 3 values (\ie $m_1=3$): $v_{1, 1}=Sunny$, $v_{1, 2}=Overcast$ and $v_{1,3}=Rain$, and they will be encoded to: $001$, $010$, and $100$, respectively. 
The last feature $s_5=Label$ has 2 values (\ie $m_5=2$): $v_{5,1}=Yes$ and $v_{5,2}=No$, and they will be encoded to $01$ and $10$, respectively. 
A data sample $D=(Sunny$, $True$, $High$, $Hot$, $No)$ will be encoded into $(001 01 01 01 10)$, consisting of 11 bits (\ie $M=11$).

Based on the bit-wise representation, we can query if a data sample contains $x$ given feature values by calculating the inner product between its encoding and a $M$-bit mask. 
Specifically, for each value $v_{i, j}$ to be queried, we set its corresponding bits in the mask to its encoding and set all the other bits to $0$. 
In this way, the inner product should be equal to $x$ if the sample contains all the $x$ values. 
In our example, if we want to check that a sample $D$ contains $(Sunny, Yes)$, the mask will be set to $(001 00 00 00 01)$, the inner product will be equal to 2 if both values are contained in $D$. 
%where $Sunny$ is the first possible value of $s_1$ and $Yes$ is the first possible value of $s_4$, we calculate the inner product between $D$ and
%check if a data sample $D$ contains \textit{Outlook=`Sunny'}, the responding mask will be ($001 00 00 00$), and the result equals to $1$ if yes. 

Here, we stress that our work focuses on training categorical features. 
Numerical features can be converted into categorical ones using methods such as discretization \cite{dougherty1995supervised}. 
Specifically, numerical values of a feature can be grouped into discrete bins.
For example, if we wanted to group the values for \textit{Temp} 2 categories this could be a possible discretization: \textit{Cool} for temperatures below $25^{\circ}C$, \textit{Hot} for temperatures equal or above $25^{\circ}C$. 
%\notewang{In Fig2, this simple example assumes Temperature has only two values (Hot and Cool), So I think here we should keep consistent (remove medium).} and $30^{\circ}C$, and \textit{Hot}  for temperature values above $30^{\circ}C$. 

\subsection{Model Representation}
\label{subsec:tree rep}
One of the main contributions of \sys is the novel way in which we represent the model as a matrix, and perform the HT training and inference as a matrix multiplication to hide the access pattern. Fig.~\ref{fig:treematrix_sub2} shows a simplified matrix representation of the model with the value expressed as strings of characters and its corresponding tree representation. Columns in the matrix map to paths of the tree. Each column contains $d-1$ elements where the $i$-th element is the value of feature $s_i$ assigned to the corresponding path\footnote{Note that the order of features in each column is fixed and same to the order defined in $S$, \ie the $i$-th value of each column must be a value of feature $s_i$. }.
In particular, if a feature $s_i$ has not been assigned to the specific path, the $i$-th element of the column is set to `$*$'. This  will be converted into specific feature values with the subsequent training. 

The last two columns in the matrix are dummy columns. In order to hide the number of tree paths from side-channel attacks, \ie the number of columns in the matrix, we add a number of dummy columns into the matrix. 
The elements in dummy columns can be of any value. More details on how dummy columns are generated will be provided in Section~\ref{sub:HT_train}.

\sloppy To make things more concrete, let's look at the example in Fig.~\ref{fig:treematrix_sub2}. The matrix representing our model consists of 4 real columns and 2 dummy columns. % with features $S=(Outlook, Windy, Humidity, Temp)$. 
The first column contains the elements $(Sunny, *, High, *)$.
This indicates that the value $Sunny$ for feature  $s_1$ (\ie $Outlook$) and value $High$ for feature $s_3$ (\ie $Humidity$) are assigned to the first path of the tree.  
Likewise, the third column $(Overcast, *, *, *)$ indicates that only the value $Overcast$ has been assigned to the third path for feature $s_1$, while the remaining 3 features have not been assigned. 
The right-hand side of Fig.~\ref{fig:treematrix_sub2} depicts the model currently stored in the matrix if it were represented as a tree.

As we said, the matrix in Fig.~\ref{fig:treematrix_sub2} is a simplified representation of how the model is stored in 
\sys.
Fig.~\ref{fig:treematrix_sub3} shows how the matrix is actually stored in the enclave as a collection of bit strings. Using the one-hot encoding technique, the matrix $\mathcal{M}_t$  only contains 0 and 1 bit. 
For instance, looking at the first column in the matrix, the values $Sunny$ and $High$ are encoded into $001$ and $01$, respectively; while the value `$*$' for feature is encoded into a string with 0 bits.

%\notegio{I am really confused about the following part. First of all, we are encoding in the matrix only the edges of the tree not the nodes. This means that there is no need to represent leaves. The other thing is that in Fig3 in the encoding part on top would be nice to see some of the examples of the $M, m_d$, etc}

Each column of matrix only contains $d-1$ values:  these are the values that could be assigned to features  excluding the values for the labels. Thus each column of $\mathcal{M}_t$ has $M-m_d$ bits, where $m_d$ is the number of values for labels.
Assuming the model has $P_{real}$ real columns and \sys inserts $P_{dummy}$ dummy columns into $\mathcal{M}_t$, the total number of columns in the matrix is   $P=P_{real}+P_{dummy}$. 
Therefore, the size of $\mathcal{M}_t$ is $(M-m_d) \times P$.

For HT training, \sys also stores the statistical information for each leaf, which is required for computing the IG value. 
In \sys, the statistical information of each leaf is stored in a 2d array $Leaf$. 
%\notegio{We do not say anything here about the dummy values in Leaf but we have them in figure 3.}
Because the number of leaves of the model should also be protected, we store in $Leaf$ some dummy values representing dummy leaves. 
Considering that each column in the model could represent a HT path with a  leaf, then $Leaf$ contains $P$ 1d arrays: $P_{real}$ arrays for real leaves and $P_{dummy}$ arrays for dummy leaves.
Precisely, the $p$-th array, $Leaf[p]$, contains all features for the $p$-th leaf, where $p \in [1, P]$. The actual values stored in $Leaf$ are the frequency values defined as $c_{(value, label)}$, for each $(value, label)$ pair. 
Note that only the $(value, label)$ pair of the features that have not been assigned to a path will be updated and used for computing IG. 
Storing the pairs of all features for all leaves ensures $|Leaf[p]|$ is the same for all leaves, which is $L=\sum_{i=1}^{i=d-1} m_i*m_d$. 
In this way, the entire model structure is protected from side-channel attacks. 
%Formally, we use $Leaf[p][l]=c_{(value, label)}$ to index the frequency of the $l$-th $(value, label)$ pair of the $p$-th leaf, where $l \in [1, L]$. 

Both $\mathcal{M}_t$ and $Leaf$ are stored within the enclave in plaintext. 

%% file: sections/ppht.tex
\section{HT Training and Inference in \sys}
\label{sec:details}
In this section, we explain how \sys obliviously trains the HT model and securely inferences unlabelled data instances. 

Although the main focus of this section is about training and inference for a single HT model, \sys can be easily extended to support a Random Forest (RF) model by performing the HT training and inference over several trees. We give the details of this extension for RF in Appendix~\ref{subsec:RF}.
%We also provide some details on how \sys handles the RF scenario. 

\subsection{Setup}
\label{sub:setup}
As the first step in the setup, the DO establishes a secure channel with an enclave instance in the CSP to share a secret key $sk$. 
For HT training and inference, all the data transmitted between the DO and the enclave will be encrypted with $sk$ and a semantically secure symmetric encryption primitive, \eg AES-GCM. 
During the setup, DO also securely shares the features $S$ and the values $V_{s_i}$ of each feature to the enclave.

\subsection{Oblivious HT Training}
\label{sub:HT_train}

\begin{figure*}[!t]
\centering
\includegraphics[width=0.9\linewidth]{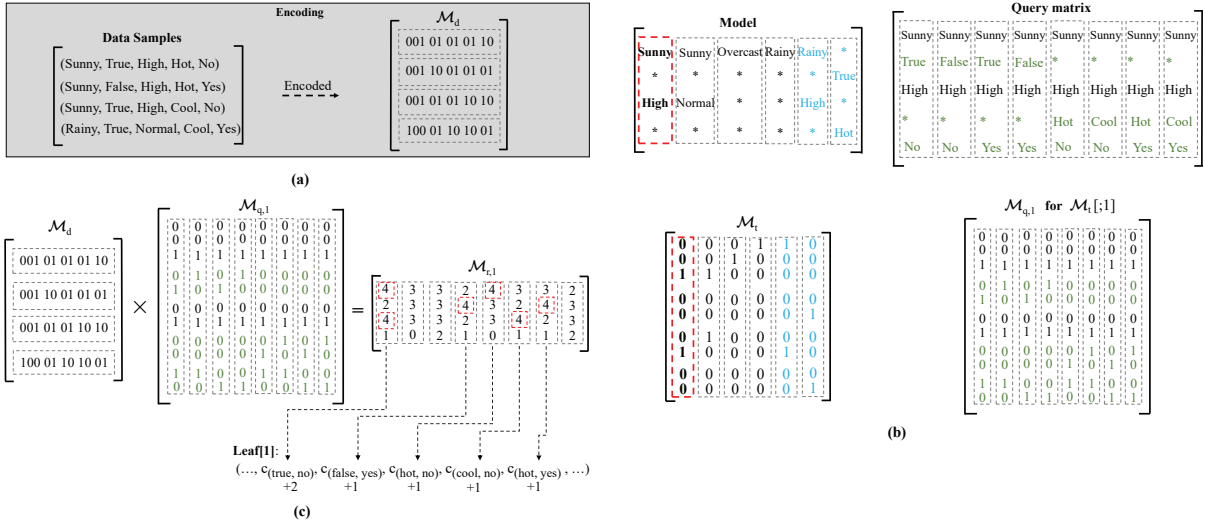}
\caption{HT training with matrix multiplication. $\mathcal{M}_d$ is the matrix of the data samples to be trained in encoded. $\mathcal{M}_q$ is the query matrix of the selected path in encoded. $\mathcal{M}_r=\mathcal{M}_d \times \mathcal{M}_q$ shows the query result. $\mathcal{M}_r[n, k]=4$ means the $n$-th data samples contains the 4 values queried by the $k$-th mask in $\mathcal{M}_q$, where $1\leq n \leq 4$ and $1\leq k \leq 8$. }
\label{fig:HT training example}
\end{figure*}

In Section~\ref{subsec:HT construction}, we have summarised the 6 steps for performing the HT training. In order to hide the tree structure during the training, the 6 steps are  modified in \sys as below:   
\begin{enumerate}
    \item classifying new arrivals into leaves with current HT model;
    
    \item updating the frequency of each $(value, label)$ pair of each feature for all leaves, not only for the leaves that receive new data samples; 
    
    \item checking if each leaf has received $n_{min}$ data samples, and performing steps 4-6 if true; 
    
    \item computing the IG value for all features, not just for unassigned features; 

    \item checking if the top two highest IG values satisfy the Hoeffding Bound; 
   
    \item if true, converting the leaf node into an internal node using the feature with the highest IG value; Otherwise, performing indistinguishable dummy operations.
\end{enumerate}

Here we present how each step is performed obliviously in \sys in details. We will use as an example the case illustrated in Fig.~\ref{fig:HT training example}.

To protect the access pattern from side-channel attacks, \sys performs the first two steps with a matrix multiplication. 
Basically, \sys converts a batch of data samples to a matrix $\mathcal{M}_d$, generates a query matrix $\mathcal{M}_q^p$ for each column $p$ in matrix $\mathcal{M}_t$, and computes $\mathcal{M}_r^p \leftarrow \mathcal{M}_d \times \mathcal{M}_q^p$. 
The elements of the resulting matrix $\mathcal{M}_r^p$ will be used to update the frequency information in the $Leaf[p]$ array. 

In the following, we take the first column of $\mathcal{M}_t$, denoted as $\mathcal{M}_t[;1]$, as an example. The different steps are shown in Fig.~\ref{fig:HT training example}. 

\noindent\textbf{Data Samples Matrix.}
%We recall here that in HT training the 6 steps are performed each time when the model receives new data samples. 
To improve efficiency, we perform the 6 steps of HT training when a batch of $N$ data samples has been stored in the  Training Buffer on the CSP. The Training Buffer stores the data samples outside the enclave. Note that the buffer size could be larger than $N$. 
When $N$ data samples are cached in the Training Buffer, \sys loads these samples into the enclave, and for each round of training, converts them into a matrix $\mathcal{M}_d$. 
Recall that \sys represents the data sample as an $M$-bit string. 
After the $N$ data samples are imported in the enclave and decrypted, \sys packs them into a $N{\times}M$ matrix $\mathcal{M}_d$, where each row of $\mathcal{M}_d$ is a data sample encoded as a bit string.
%Fig.~\ref{fig:HTtraining_sub1}
Fig.~\ref{fig:HT training example}a shows an example where $N=4$, and each data sample is represented as a $11$-bit string. 
Thus, the resulting size of the matrix \textbf{$\mathcal{M}_d$} is $4{\times}11$.

\noindent\textbf{Query Matrix.}
%In the original HT algorithm, the motivation of the first 2 steps is to query whether any of data samples (i) belong to a leaf, \ie contain the branch values on the path, and (ii) contain a $(value, label)$ pair of the features that have not been assigned to the path. For instance, for the left-most leaf of the tree in Fig.~\ref{fig:sub1}, 
Assume column $\mathcal{M}_t[;p]$ contains $\tau_p$ assigned feature values and $u_p$ unassigned features.
Our next step is to query whether any data sample in the current batch contains (i) the $\tau_p$ feature values assigned in the column $\mathcal{M}_t[;p]$, and (ii) a $(value, label)$ pair for any of the $u_p$ features that are not still assigned.

We perform this query by means of a matrix multiplication and the result of this multiplication will be another matrix $\mathcal{M}_r^p$. 
% $\mathcal{M}_t[;p]$
The elements in $\mathcal{M}_r^p$ are then used to update the frequencies of the queried $(value, label)$ pairs in the array $Leaf[p]$.

The process of generating a query matrix $\mathcal{M}_q^p$ for a given column $\mathcal{M}_t[;p]$ is then reduced to define a set of $M$-bit masks which form the columns in $\mathcal{M}_q^p$. 
Each mask can only check one case. 
%all possible $(value, label)$ pairs of unassigned features.
Thus the number of masks, \ie the number of columns of matrix $\mathcal{M}_q^p$, is determined by the possible values of unassigned features and the possible values of the label. 
In more detail, for $\mathcal{M}_t[;p]$, $\mathcal{M}_q^p$ are determined by (i) the $\tau_p$ assigned feature values (these will be the same across all the column of the query matrix); and (ii) all the possible combinations of the $(value, label)$ pairs for the $u_p$ unassigned features. 

%Fig.~\ref{fig:HTtraining_sub2}
To make things more concrete, let us look at Fig.~\ref{fig:HT training example}b, where both the model matrix $\mathcal{M}_t$ and the query matrix $\mathcal{M}_q^1$ for column $\mathcal{M}_t[;1]$ are presented in human-readable and bit-string forms. 
As we can see from the figure, $\mathcal{M}_t[;1]$ includes 2 assigned feature values (\ie $Sunny$ and $High$), and 2 unassigned features (\ie $Windy$ and $Temp$). This means that $\tau_1=2$ and $r_1=2$.

The number of columns (\ie masks) in $\mathcal{M}_q^p$ is defined as $L'=\sum_{i=1}^{i=u_p}m_i*m_d$ where $m_i$ are the possible values of each unassigned feature and $m_d$ are the possible values of the $Label$. This means that the size of $\mathcal{M}_q$ is $M\times L'$. 

%Fig.~\ref{fig:HTtraining_sub2}
In the example in Fig.~\ref{fig:HT training example}b, as both the unassigned features, $Windy$ and $Temp$, and the label $Label$ have 2 possible values (\ie  $V_{Windy}=(True, False)$, $V_{Temp}=(Hot, Cool)$, and $V_{Label}=(Yes, No)$), the total number of masks that we need to query is given by the following:   $|V_{Windy}|*|V_{label}|+|V_{Temp}|*|V_{label}|=8$. In other words,  for column $\mathcal{M}_t[;1]$ we need a query matrix $\mathcal{M}_q^1$ of 8 columns with the values for each column shown in  Fig.~\ref{fig:HT training example}b. 

%In Fig.~\ref{fig:HT training example}, the masks for querying the 8 types of samples of the left-most leaf are the 8 columns of $\mathcal{M}_q$.
    
    % \item \textbf{How to define each mask:} As mentioned in Section~\ref{subsec:sample}, the mask is defined based on the value to be queried. Specifically, for each value to be queried, we set its corresponding bits in the mask to its encoding and set all the other bits to 0. 
    % For a certain path, the values on the branches have been determined, thus their corresponding bits in the $L'$ masks should be fixed to their encoding. For instance, in Fig.~\ref{fig:HT training example}, the left-most path has two branch values: $Sunny$ and $High$, thus the encoding bits corresponding to $Outlook$ and $Humidity$ features are $001$ and $01$, respectively. 
    % %\notegio{This sentence needs to be checked carefully}
    % \revision{For the mask querying $(v_{i, \mu}, v_{d, \nu})$, where $1 \leq \mu \leq m_i$ and $1 \leq \nu \leq m_d$, the corresponding bits should be set to the encoding of $v_{i, \mu}$ and $v_{d, \nu}$}, respectively, and all the remaining bits of the mask are set to $0$. 

\noindent\textbf{Matrix multiplication.}
By computing $\mathcal{M}_d \times \mathcal{M}_q^p$, we get a $N \times L'$ result matrix $\mathcal\mathcal{M}_r^p$. 
We use $\mathcal\mathcal{M}_r^p[n, k]$ to represent its element at the $n$-th row and $k$-th column, where $n \in [1, N]$ and $k \in [1, L']$. 
$\mathcal\mathcal{M}_r^p[n, k]$ is the inner product between the $n$-th data sample and the $k$-th mask.
This value represents the number of values in $n$-th data sample that match the values in the $k$-th column of the query matrix. 
We are interested in finding the data samples that fully match the values defined in $\mathcal{M}_q^p[;k]$: the $\tau_p$ assigned feature values in $\mathcal{M}_t[;p]$ and the $(value, label)$ pair that we are querying for.  
In other words, if $\mathcal{M}_r[n, k] =\tau_p+2$ the $n$-th sample matches the mask $\mathcal{M}_q^p[;k]$. 
To be more concrete, let us look at a specific case presented in Fig.~\ref{fig:HT training example}c. 
Recall that we are querying for $\mathcal{M}_t[;1]$: this column has two fixed values $Sunny$ and $High$. Thus we are looking for a matching value of $\tau_1+2=4$.  
In Fig.~\ref{fig:HT training example}c, we can see all the elements $\mathcal{M}_r^1[n, k]=4$ highlighted in red boxes.

The next step is to update the frequency information of each $(value, label)$ pair contained in the $Leaf$ arrays. 
This is performed by scanning each column of the result matrix $\mathcal\mathcal{M}_r^p$ and checking how many elements in each column is equal to $\tau_p+2$. 
For instance, in Fig.~\ref{fig:HT training example}c, the first column of $\mathcal{M}_r^1$ contains two matches. The corresponding frequency value $c_{(True, No)}$ in $Leaf[1]$ is increased by 2. 
Here the enclave uses a mapping $\sigma$ to map the columns of $\mathcal{M}_r^q$ to the elements in $Leaf[p]$.
%$l={\sigma}(1)$ and $\sigma(1)$ locates $\mathcal{M}_{r,1}$ to the position of $Leaf[1]$.}
%thus the frequency value $c_{(True, No)}$ for $Leaf[1]$ is increased by 2.

\sys executes these operations obliviously, otherwise an adversary could use side-channel attacks to learn which data sample contains which pair.  
Precisely, \sys linearly scans each column of $\mathcal\mathcal{M}_r^p$, using $\mathtt{oequal}$ to check how many elements in $\mathcal\mathcal{M}_r^p[;k]$ equal to $\tau_p+2$.
At the last step, the frequency counts are added to the corresponding $c_{(value, label)}$ value in the relevant leaf array using $\mathtt{oassign}$. 

%Note that when there is no unassigned feature left in the leaf, this leaf will not be converted into an internal node, and the enclave just needs to record the frequency of each label in this leaf. In this case, since only the $label$ and $\tau$ branch values should be queried, the inner product is equal to $\tau+1$. The enclave needs to obliviously check how many elements in each column of $\mathcal{M}_r$ are $\tau+1$, and obliviously adds these values to the corresponding $c_{label}$. 

During the training, \sys requires to access all the columns of $\mathcal{M}_t$ and generate a query matrix for each column. 
Even if this operation is executed in the enclave, with side-channel attacks, an adversary could infer information about the model (\eg the number of columns, which maps to the number of HT paths). 
Likewise, when accesses are made to $Leaf$ for updating the frequency information, the adversary could also infer the number of leaves. 
To prevent such a leakage, \sys inserts dummy columns and dummy arrays into $M_{t}$ and $Leaf$ during the setup. \sys uses a $P$-bit string $isDummy$ to mark if $\mathcal{M}_t[;p]$ and $Leaf[p]$ is real or dummy. 
%For the whole tree, the enclave indeed just performs only one matrix multiplication by combining the query matrix of each leaf together. As \sys represents the tree as matrix $\mathcal{M}_t$, the query matrix of a path  \revision{is generated from a column of .} 
%\notegio{I do not understand what "extended" means in this context: is indeed extended from a column of $\mathcal{M}_t$.} For generating the query matrix for the whole tree, the enclaves needs to access all the columns of $\mathcal{M}_t$, from which the adversary could infer the number of columns (\ie the number of paths of the tree) with side-channel attacks, \eg timing attacks.  

\begin{figure}[!ht]
  \centering
      \includegraphics[width=\linewidth]{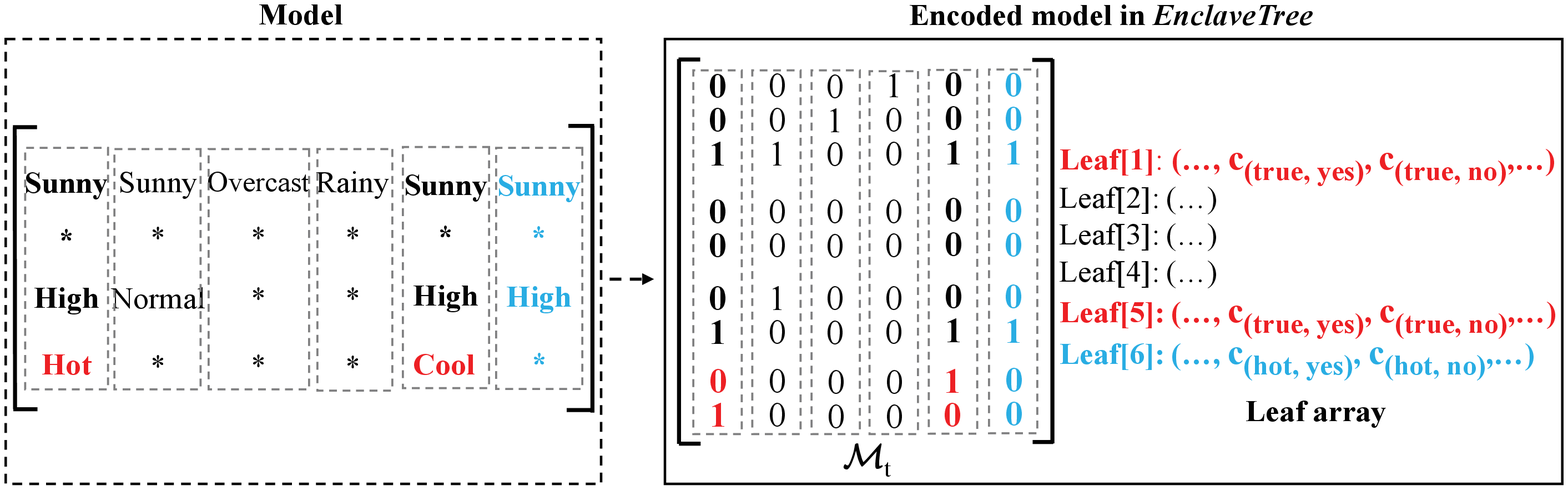}
        \caption{The model after one round training. The parts set in red are those modified after one round of training.}
        \label{fig:matrix after training}
\end{figure}

\noindent\textbf{Oblivious model construction.}
Once the frequency of each pair has been updated, the IGs of those leaves that have received $n_{min}$ data samples can be securely computed within the enclave.
%\notewang{Note that, in step 3, we argue that checking if a leaf has received  $n_{min}$ data samples does not leak any private information. For step 4, \sys calculates the IGs of all features to hide the number of assigned features.}
However, the last 2 steps should be performed obliviously as they involve memory access. 

For step 5, the enclave uses $\mathtt{ogreater}$ and $\mathtt{oequal}$ to obliviously find out the two features with the highest IG values for each leaf. 
Assume the two features are $s_a$ and $s_b$ for $Leaf[p]$, where $\overline{G}(s_a) > \overline{G}(s_b)$. The enclave uses $\mathtt{ogreater}$ to check if $\overline{G}(s_a)-\overline{G}(s_b)>{\epsilon}$. 
If true, the enclave selects the feature $s_a$ using $\mathtt{oselect}$ and performs the last step, \ie converting $Leaf[p]$ into an internal node with $s_a$.

In terms of the tree structure, converting a leaf into an internal node means assigning $s_a$ to the leaf, outputting $m_a$ branches with $m_a$ new leaves, and assigning the $m_a$ values of feature $s_a$ to the new branches. 
In terms of the matrix model in \sys, the enclave modifies $\mathcal{M}_t$ and $Leaf$ with the following extensions.

\textbf{$\mathcal{M}_t$ extension:} 
    To hide whether the model is extended after each round of training, \sys converts $m_a-1$ dummy columns into real ones by resetting $isDummy$, rather than adding new columns into $\mathcal{M}_t$. In more details, \sys first copies the values of $\mathcal{M}_t[;p]$ to $m_a-1$ dummy columns, and then assigns the $m_a$ values of feature $s_a$ to $\mathcal{M}_t[;p]$ and the $m_a-1$ dummy columns with $\mathtt{oassign}$. 
    Fig.~\ref{fig:matrix after training} shows how $\mathcal{M}_t$ is changed when $Leaf[p]$ is converted into an internal node with feature $Temp$. In the example, $m_{Temp}=2$, thus only one dummy column, $\mathcal{M}_t[;5]$, is converted into a real one. The last 2 bits of $\mathcal{M}_t[;1]$ and $\mathcal{M}_t[;5]$are changed to 01 and 10,  respectively (the encoding for $Hot$ and $Cool$, respectively).   
    
\textbf{$Leaf$ extension:} 
    As $m_a$ new leaves are added, the leaf array $Leaf$ should also be updated. 
    Similarly, \sys first converts $m_a-1$ dummy arrays into real ones by initializing all the possible $c_{(value, label)}$ of unassigned features to $0$. The original leaf $Leaf[p]$ will be used to store the statistical information of the new $p$-th leaf, and its each $c_{(value, label)}$ is set to $0$.

During the setup, the enclave generates a number of dummy columns and leaves in $\mathcal{M}_t$ and $Leaf$, respectively. 
As dummy values in both the model and the $Leaf$ arrays are processed as real values, a large number of dummies will degrade the performance. 
To balance efficiency with security, \sys periodically generates new dummies. 
In detail, after $\gamma$ extensions, \sys checks the number of remaining dummy values, and if this value is below a given threshold $T$, \sys generates new dummies. 
The threshold $T$ should ensure there are enough dummies for $\gamma$ extensions. 
In the worst case, all of the $\gamma$ leaves are split and generate $\gamma*(m_{max}-1)$ new leaves, where $m_{max}=\max\{m_1, ..., m_{d-1}\}$. 
We thus set $T=\gamma * (m_{max}-1)$.  

% \subsection{Oblivious Predication}
\subsection{Oblivious HT Inference}
\label{sub:HT_inference}
\begin{figure}[!t]
\centering
\includegraphics[width=\columnwidth]{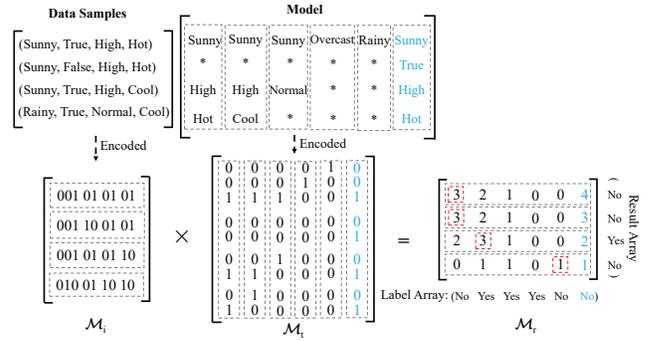}
\caption{HT inference with matrix multiplication.}
\label{fig: HT inference example}
\end{figure}
One of the features of data stream classifications is that unlabelled data instances can be received for inference at any time. In other words, there is not a clear separation between a training and an inference phase. As such, \sys has to be able to support  inference operations while the model is being trained. 
%\notewang{This is also the reason that existing cryptography solutions are not suitable to data stream classification.}

The target of HT inference is to return a classifying label value for each data instance to the DO. 
Before classifying any data instance, \sys has to define the label values in the current model. 
Data samples with different label values could be classified into the same leaf during the training. The label value with the highest frequency will be used as the label value of the leaf. 
For the $p$-th leaf, the label value that has the highest frequency can be obtained by checking the $c_{(value, label)}$ in $Leaf[p]$ with oblivious primitives. 

To protect the enclave access pattern, \sys also performs the HT inference with a matrix multiplication. 
In more detail, the Oblivious Inference sub-component of \sys processes a batch of instances each time. 
Assume the batch size for HT inference is $N'$. 
After loading and decrypting $N'$ data instances, \sys converts the instances into a matrix $\mathcal{M}_i$. 
\sys also represents each data instance with a bag of bits. Compared with data samples, the bit string of a data instance only has $M-m_d$ bits as the data instance does not have label values.
Thus, the size of $\mathcal{M}_i$ is $N' \times (M-m_d)$. 
For instance, in Fig~\ref{fig: HT inference example}, each column of $\mathcal{M}_i$ has 9 bits.

\sys performs the inference by computing $\mathcal{M'}_r \leftarrow \mathcal{M}_i \times \mathcal{M}_t$. 
Since the size of $\mathcal{M}_i$ and $\mathcal{M}_t$ are $N' \times (M-m_d)$ and $(M-m_d) \times P$ respectively, the size of $\mathcal{M'}_r$ is $N' \times P$. 
The element $\mathcal{M'}_r[n, p]$ indicates whether the $n$-th data instance belongs to the $p$-th path, where $n \in[1, N']$. 
If this is the case, then $\mathcal{M'}_r[n, p]=\tau_p$. %, where $\tau_p$ is the number of assigned feature values in the $p$-th column. 
$\tau_p$ can be easily obtained by checking how many $1$ bits\footnote{The one-hot encoding ensure that the encoded value for each feature has only one bit set 1.} are in the $p$-th column of $\mathcal{M}_t$.

To check which path the $n$-th data instance belongs to, the enclave scans the $n$-th row of $\mathcal{M'}_r$ and checks if $\mathcal{M}_t[n, p]=\tau_p$ with $\mathtt{oequal}$. 
If this is true, then the label value of the $p$-th leaf will be the inference result for the $n$-th data instance. 
Finally, the enclave encrypts the $N'$ labels with $sk$ and sends them to the DO.

%% file: sections/implementation.tex
\begin{figure*}[ht!]
  \centering
    \begin{subfigure}{0.49\textwidth}
      \centering   
      \includegraphics[width=0.85\linewidth]{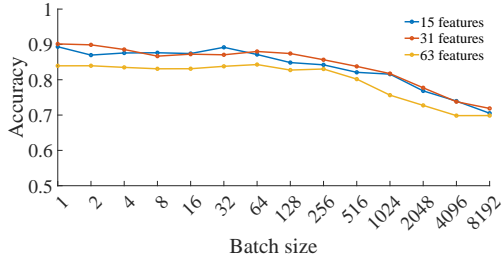}
        \caption{HT accuracy comparison across different batch size}
        \label{fig9:sub1}
    \end{subfigure}   %      \hfill  % 
    \begin{subfigure}{0.49\textwidth}
      \centering   
      \includegraphics[width=0.85\linewidth]{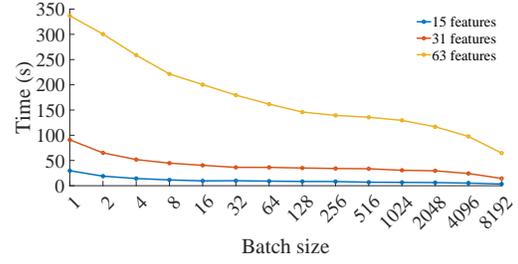}
        \caption{HT runtime comparison across different batch size}
        \label{fig9:sub2}
    \end{subfigure}
\caption{
\label{fig:HT_acc_batchRuntime}
The performance of HT training with different batch size
}
\end{figure*}

\begin{figure*}[!ht]
  \centering
    \begin{subfigure}{0.49\textwidth}
      \centering   
      \includegraphics[width=0.85\linewidth]{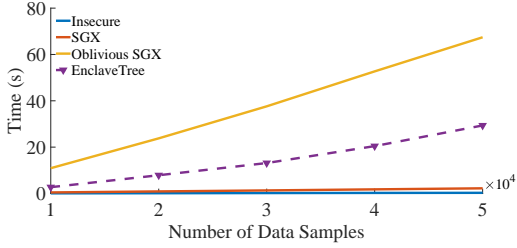}
        \caption{HT training across different number of samples}
        \label{fig7:sub1}
    \end{subfigure}   %      \hfill  % 
    \begin{subfigure}{0.49\textwidth}
      \centering   
      \includegraphics[width=0.85\linewidth]{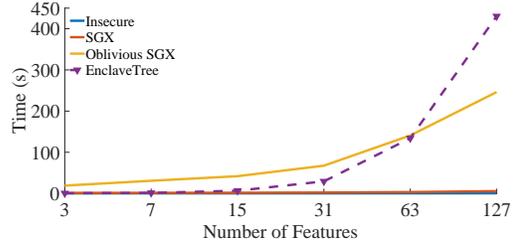}
        \caption{HT training across different number of features}
        \label{fig7:sub2}
    \end{subfigure}
\caption{
\label{fig:HT_training_perf}
The performance of HT training under different settings
}
\end{figure*}

% \begin{figure*}[ht!]
%   \centering
%     \begin{subfigure}{0.49\textwidth}
%       \centering   
%       \includegraphics[width=0.8\linewidth]{graph/exp/HT_evaluation_31.eps}
%         \caption{HT inference across different depth-leaves pairs}
%         \label{fig8:sub1}
%     \end{subfigure}   %      \hfill  % 
%     \begin{subfigure}{0.49\textwidth}
%       \centering   
%       \includegraphics[width=0.8\linewidth]{graph/exp/HT_evaluation_DiffFea_Comp.eps}
%         \caption{HT inference across different number of features}
%         \label{fig8:sub2}
%     \end{subfigure}
% \caption{
% \label{fig:HT_inference_perf}
% The performance of HT inference under different settings
% }
% \end{figure*}

\begin{figure}[!ht]
  \centering
      \includegraphics[width=0.85\linewidth]{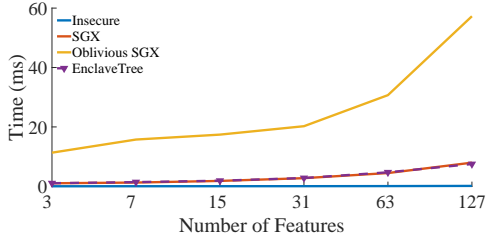}
\caption{HT inference across different number of features}
\label{fig8:sub2}
\end{figure}

\section{Implementation and Evaluation Results}
\label{sec: implementation}
In this section, we first describe the implementation of \sys. We then describe the evaluation test-bed we used for running our experiments. % and chosen cryptographic parameters. 
Finally, we conclude this section with a detailed performance analysis. 
%for all the components to highlight the individual contributions of the techniques described in the previous sections.

\subsection{Implementation}
The prototype of \sys is implemented in C++ based on the machine learning library \textit{mlpack} \cite{mlpack2018}. Mlpack implements the original HT algorithm (also known as Very Fast Decision Tree, VFDT) given in \cite{domingos2000mining}. We modify both the training and inference into matrix-based processes according to our approach. To make the algorithm oblivious, we implemented  oblivious primitives with inline assembly code (as done in \cite{ohrimenko2016oblivious, law2020secure, poddar2020visor}). %Most of these codes comprise the x86-64 $\mathtt{cmp}$ and $\mathtt{cmovz}$ instructions and require AVX2 vector instructions to speed up $\mathtt{oaccess}$ primitive.

\subsection{Experiment Setup}

\noindent\textbf{Testbed.}
We evaluated the prototype of \sys on a desktop with AVX2 and SGX support, where AVX2 feature is required for $\mathtt{oaccess}$. 
The desktop contains 8 Intel i9-9900 3.1GHZ cores and 32GB of memory (${\thicksim}$93 MB EPC memory), and runs Ubuntu 18.04.5 LTS and OpenEnclave 0.16.0. 
%Our libraries and code are compiled using clang 10.0.1 with the O3 flag and built against GNU C Library in the host and MUSL C Library in the enclave respectively.

\noindent\textbf{Baselines.} 
%\notegio{NOTE that I have named each of the baselines and these names should be used throughout the paper.}
To the best of our knowledge, there is no other approach in the wild that can be used for a performance comparison with \sys. Therefore, to better show the performance of \sys, we implemented and evaluated 3 baseline cases named $\mathtt{Insecure}$, $\mathtt{SGX}$, and $\mathtt{Oblivious~SGX}$. 
$\mathtt{Insecure}$ baseline does not provide any protection and performs the traditional HT training and inference in plaintext where each data sample is classified level by level from the root to a leaf node \cite{domingos2000mining}. Note that this baseline is performed  without using SGX enclaves and in plaintext therefore it does not provide any security. $\mathtt{SGX}$ baseline performs the traditional HT training and inference within an enclave but without protecting the access pattern. 
By comparing the performance of the first two baselines, we can see the overhead incurred by using SGX. 
To protect the enclave access pattern, $\mathtt{Oblivious~SGX}$ baseline obviously performs the traditional HT training and inference within the enclave with oblivious primitives. 
We leverage the strategy used in~\cite{law2020secure} for implementing $\mathtt{Oblivious~SGX}$, where the nodes of each level are stored in an array and the target node is obliviously accessed with $\mathtt{oaccess}$. Moreover, dummy nodes are generated to hide the real number of nodes in each level.
 
%\notegio{the name for the last baseline is not correct to me: it is basically like the second baseline but with the addition of the oblivious primitives} 
%In contrast, \sys uses matrix multiplications. % instead of traditional HT training within an enclave and oblivious primitives. 
When outside the enclave, the data samples are encrypted with 128-bit AES-GCM in $\mathtt{SGX}$, $\mathtt{Oblivious~SGX}$, and \sys. 

All the experiment results presented in the following are average over 100 runs. 

\noindent\textbf{Batch size.} 
The batch size $N$ affects the performance of HT training and also the inference accuracy.  
As shown in Fig.~\ref{fig:HT_acc_batchRuntime}, the performance of HT training improves at  the increase of $N$, whereas the accuracy of the model decreases with the increase of $N$. 
\sys processes each batch of data samples in one step, which means the 6 steps of the HT training are performed once every $N$ data samples. 
As a result, less computation is required when $N$ gets larger, yet the best moment to covert leaves to internal nodes could be missed. 
%}\notegio{Fig5 does not say what 15,31 and 63 represent. Actually we do not say even in the text here. This needs to be addressed.}
From Fig.~\ref{fig:HT_acc_batchRuntime}, we can also notice that when $N < 128$ the accuracy of the model decreases very slightly (Fig.~\ref{fig:HT_acc_batchRuntime}a) but the decrease in runtime overhead is much more dramatic especially when considering 63 features (Fig.~\ref{fig:HT_acc_batchRuntime}b). 
For $N=100$, the accuracy of the model is almost the same as for $N=1$.
Thus, in the following experiments, we set $N=100$. 
%\notegio{The question any reviewer would have is: why not 128 then? We need to explain the value of N better.}

% \subsection{The Overhead on DO}
% % We have HE-based client cost as well, put it here?
% % 31 fea: 13.47KB/100 instances 0.13KB/1 instance 51.37ms for encrypting 10000 samples.
% \revision{We measured the overhead on DO when it encrypts the data samples with 128-bit AES-GCM. 
% In our experiment, the DO performs the encryption over every 100 data samples, which only takes 0.48 ms and 0.13 KB memory.}

\subsection{Evaluation on Real Datasets}
\begin{table}[!t]
	\centering
	\footnotesize
	\caption{Datasets}
	\label{tab:dataset}
	\begin{tabular}{cccc}
		\toprule
		\textbf{Dataset} & \textbf{\#Features} & \textbf{\#Labels}  & \textbf{\#Samples} \\	\hline
		Adult     & $14$ & $2$ & 32,561  \\	%\hline
		REC       & $9$  & $2$ & 5,749,132  \\	%\hline
% 		Covertype & $45$ & $7$ & 581,012 \\	\bottomrule
        Covertype & $54$ & $7$ & 581,012 \\	\bottomrule
	\end{tabular}
\end{table}
\begin{table}[!t]
	\centering
	\footnotesize
	\caption{Training runtime on real datasets (s)}
	\label{tab:performance_real}
	%\begin{ThreePartTable}
	\begin{tabular}{cccc}
		\toprule
		\textbf{Scheme} & \textbf{Adult}  & \textbf{REC}  & \textbf{Covertype} \\
		\hline
		$\mathtt{Insecure}$ & 0.09 & 16.06 & 118.29  \\
		%\hline
		$\mathtt{SGX}$ & 0.87 & 98.11 & 773.50  \\
		%\hline
		$\mathtt{Oblivious~SGX}$ & 33.04 & 909.24 & 1772.77 \\
		%\hline
		\textbf{\sys} & 19.05 & 128.28 & 6930.51 \\
		\bottomrule
	\end{tabular}
	%\end{ThreePartTable}
\end{table}

We first evaluated the performance of HT training with 3 real datasets that are widely used in the literature: Adult dataset, Record Linkage Comparison Patterns (REC) dataset, and Covertype dataset. They are obtained from UCI Machine Learning Repository~\footnote{https://archive.ics.uci.edu/ml/}. 
The details of each dataset are shown in Table~\ref{tab:dataset}.
In particular, we use the Adult and REC datasets to evaluate the performance of HT training, and use the REC dataset to test the performance of RF training, where 100 trees are trained and each tree consists of 7 features. 
The results are shown in Table~\ref{tab:performance_real}. 
For Adult and REC, \sys outperforms $\mathtt{Oblivious~SGX}$ by ${\thicksim}1.7{\times}$ and ${\thicksim}7.1{\times}$, respectively. 
%Compared to the ${\thicksim}5.7{\times}$ speedup in $3{\times}10^4$ samples with 15 features, \sys encounters a performance drop is because some features in the Adult dataset have multiple categories rather than two. 
When training with Covertype, \sys's performance is worse than  $\mathtt{Oblivious~SGX}$.
This is because a large query matrix $\mathcal{M}_q$ is required to process the 54 features in Covertype. 
%\sys is more suitable to process the data streams for the scenarios with about a dozen of features.

%two feature values in the REC dataset have continuous values. We convert them to binary features by computing their percentile and encoding the values according to their percentile range. 
%For instance, \notewang{a continuous feature in REC takes value in [0,1], we select a threshold from the range to convert it into binary features}. 
%The Covertype dataset contains tree observations from four areas of a forest and is often used to train RF. 
%55 features
%We use only 45 of these features that are binary. 

\subsection{Performance of HT Training and Inference}
\label{subsec:HT_eva}
We also evaluated the performance of \sys with synthetic datasets which allow us the flexibility to change the number of data samples and features to better show the performance of \sys under different conditions. The machine learning package, \textit{scikit-multiflow} \cite{skmultiflow}, is employed to generate the streaming data samples in our test. 
The main operation involved in \sys is matrix multiplication, the performance of which is affected by the matrix size and determined by the number of features and number of values of each feature. 
In the following test, we set the number of values to $2$ for all features and modify the matrix size by changing the number of features. 
From the above 3 datasets, we can see that the datasets usual contain dozens of features. 
The observation presented in~\cite{nie2010efficient,wang2020discriminative} also shows that dozens of features, \eg 10, 20, or 30, are usually enough to reflect the distribution of the dataset. 
However, to better analyze the performance of \sys, in our tests we set the number of features to range between $3$ and $127$.

% \notegio{The next paragraph is not good enough. I am thinking of removing most of it. The main point is that our system is practical as long as it is able to process the arriving data without creating too much of a backlog. The question is what is a practical frequency of arrival for our system? As a way of an example: if the buffer size is 100 and to process those 100 samples for training takes 10 seconds, ideally then the data arrival time should be 100/10 data items per second, which means 10 samples per second. Can we have some details on this?}
% For training and inferring data streams, whether the efficiency of \sys matters depends on the rate of receiving data samples or instances.  
% When the data arrive with a low rate, \eg one data sample per hour, \sys is practical as long as the training and inference is finished before the next data arrives. 
% %the buffers in the CSP to cache labelled and unlabelled data samples have limited size, \notegio{What happens when the buffer capacity is exhausted?}
% However, when the data samples arrive with an extremely high rate, \eg thousands of samples per second, the training and prediction should be efficient so as to avoid congestion. 
% To evaluate the performance of \sys, we consider the second situation, where a large volume of data samples have cached in the CSP buffer to be trained or predicated. \notegio{We do not way what this high volume rate is.}
\begin{table}[!t]
	\centering
	\footnotesize
	\caption{The performance of HT inference.}
	\label{tab:inference_info}
	 \renewcommand\tabcolsep{2.0pt} 
	\begin{tabular}{c|cccc}
		\toprule
		{\textbf{\#Data samples}} & \multicolumn{4}{c}{\textbf{HT inference runtime (ms)}} \\	%\cmidrule(lr){4-7}
		 & $\mathtt{Insecure}$ & $\mathtt{SGX}$ & $\mathtt{Oblivious~SGX}$ & \textbf{\sys} \\ \hline 
		$1{\times}10^4$ & 0.05 & 2.79 & 9.78 & 2.36 \\
		%\hline 
		$2{\times}10^4$ & 0.05 & 2.89 & 10.56 & 2.43\\
		%\hline 
		$3{\times}10^4$ & 0.05 & 2.86 & 14.34 & 2.59\\
		%\hline 
		$4{\times}10^4$ & 0.05 & 2.85 & 15.19 & 2.62\\
		%\hline 
		$5{\times}10^4$ & 0.05 & 2.89 & 20.24 & 2.80\\
		\bottomrule
	\end{tabular}
\end{table}
\noindent\textbf{Performance of HT training.}
%In our experiment, the HT training is performed with only one multiplication by combing the query matrix of each column into a larger one. 
To measure the training performance of \sys, we performed two sets of experiments: 1) first we fixed the number of features while we changed the number of data samples; and 2) we fixed the number of data samples while we changed the number of features. 

In the first set of experiments, we set the number of features to 31, which is large enough to cover most of the data stream scenarios, and changed the number of data samples from $1{\times}10^4$ to $5{\times}10^4$ samples.
In the second test, we fixed the number of data samples to $5{\times}10^4$ and increased the number of features from 3 to 127. 
For the same settings, we compare the performance of \sys with the other three baselines and the results are presented in Fig.~\ref{fig:HT_training_perf}.

Fig.~\ref{fig7:sub1} shows the execution time in seconds to perform the training with fixed features. 
From the results we can see that \sys needs less time than $\mathtt{Oblivious~SGX}$ but more time than $\mathtt{SGX}$.
Precisely, \sys outperforms $\mathtt{Oblivious~SGX}$ by $4.03{\times}$, $3.02{\times}$, $2.86{\times}$, $2.58{\times}$, $2.29{\times}$, but incurs ${\thicksim}6{\times}$, ${\thicksim}9{\times}$, ${\thicksim}10{\times}$, ${\thicksim}11{\times}$, ${\thicksim}13{\times}$ overhead for protecting the access pattern when compared to $\mathtt{SGX}$ for the five cases, respectively.

%\notewang{Figure 8b, $33.404{\times}$, $23.318{\times}$, $6.230{\times}$, $2.297{\times}$, $1.050{\times}$}
Fig.~\ref{fig7:sub2} shows the results when we fix the data sample size and vary the number of features. As expected, the training time increases with the increase of the number of features. 
It is interesting to note that for less than 63 feature, \sys execution time is better than $\mathtt{Oblivious~SGX}$. However, with more than 63 features, $\mathtt{Oblivious~SGX}$ outperforms \sys in terms of execution times. 
The main reason of this increase in execution time is the increase in size for the matrices $\mathcal{M}_q^p$, $\mathcal{M}_d$ and $\mathcal{M}_r^p$. These matrices become larger at the increase of the number of features, and this increases the running time for performing the matrix multiplication to get $\mathcal{M}_r$. 
%Moreover, larger matrices requires more memory exhausting the EPC size. As a result, the enclave has to resorts to expensive EPC paging which drastically reduces the performance. 

\noindent\textbf{Performance of HT inference.}
%The execution time for HT inference depends on the structure of the tree. In particular, the depth of the tree and the number of paths are determined by the data samples used for training.  
To evaluate the performance of inference, we also conducted two sets of experiments: 1) first, we fixed the number of features to $31$ and changed the number of data samples from $1{\times}10^4$ to $5{\times}10^4$; 
and 2) then we fixed the data samples to $5{\times}10^4$ and changed the number of features from $3$ to $127$.
In both sets of experiments, we set the batch size $N'=100$, \ie $100$ data instances are classified with one matrix multiplication. 
The results for both experiment sets are shown in Table~\ref{tab:inference_info} and Fig.~\ref{fig8:sub2}, respectively \footnote{We also provide the results with 15 and 63 features in Appendix~\ref{sub:HT_performance_appendix}.}.

From both Table~\ref{tab:inference_info} and Fig.~\ref{fig8:sub2}, we can see that despite being the most secure of all the other baselines, the HT inference in \sys is very comparable to that of $\mathtt{SGX}$ (\sys performance is even better than $\mathtt{SGX}$ in some cases). 
The results also show that \sys is faster than $\mathtt{Oblivious~SGX}$ (up to ${\thicksim}7.23{\times}$ times). %\notewang{(up to ${\thicksim}7.23{\times}$ times).}
%\notegio{Which case are we referring to?}). 
%\notewang{This is for 31 features in Fig 7a. For other features, like 7, it could be up to ${\thicksim}17{\times}$}

%% file: sections/related.tex
\section{Related Work}
\label{sec: related_work}
In this section, we review existing privacy-preserving approaches for general ML algorithms and for data stream classification.   

\subsection{Privacy-preserving Machine Learning} 

\noindent\textbf{Cryptography-based Solutions.} 
Most of the existing privacy-preserving works~\cite{lindell2000privacy,du2002building,vaidya2005privacy,wang2006classification,xiao2005privacy,emekcci2007privacy,samet2008privacy,de2014practical,akavia2019privacy,liu2020towards,akavia2019privacy,liu2020towards} rely on cryptographic techniques, such as SMC and HE. 
Compared with \sys, these schemes require multiple rounds of interaction between different participants. 
The schemes proposed in  \cite{lindell2000privacy,du2002building,vaidya2005privacy,xiao2005privacy,emekcci2007privacy,samet2008privacy} leak the statistical information and/or tree structures to the CSP. 
Moreover, as shown in \cite{ohrimenko2016oblivious}, these cryptographic solutions incur heavy computational overheads. None of these works is suitable for data stream classification.
%\notewang{In addition, Table 1 is the comparison of complexity. We have the Client Encryption Time (AES), which is definitely faster.}

%Various works \cite{bost2015machine,wu2016privately,de2017efficient,tueno2019private,zheng2019towards,liu2020towards} propose to address the problem of privacy-preserving inference relying on HE and/or SMC.
%As shown in \cite{ohrimenko2016oblivious}, these cryptographic solutions show  severe slowdowns when compared to TEE-based methods. 

\noindent\textbf{TEE-based Solutions.} 
In recent years, advances in TEE technology have enabled a set of exciting ML applications such as Haven \cite{baumann2015shielding} and VC3 \cite{schuster2015vc3}. However, TEE solutions (\eg Intel SGX) are vulnerable to a large number of side-channel attacks. 
%, such as cache-timing attacks \cite{brasser2017software}, branch prediction attacks \cite{lee2017inferring} and page monitoring \cite{xu2015controlled}. 
Decision tree is vulnerable to those attacks as it induces data-dependent access patterns when performing training and inference tasks inside the enclave. 
Raccoon \cite{rane2015raccoon} proposes several mechanisms for data-oblivious execution for TEE to prevent these attacks. 
Ohrimenko et al. \cite{ohrimenko2016oblivious} propose to make the decision tree inference oblivious with oblivious primitives.  
Motivated by \cite{ohrimenko2016oblivious}, Secure XGBoost \cite{law2020secure} makes both the XGBoost model (a variant of the decision tree) training and inference oblivious with oblivious primitives. 
Combing TEE with oblivious primitives can prevent side-channel attacks and achieve better performance than cryptographic-based solutions. 
However, the use of oblivious primitives still leads to prohibitive performance overheads. 
\sys significantly reduces the need of using oblivious primitives because the access pattern to the model is hidden by the use of matrix multiplication. 
We only use oblivious primitives to process the results of the result matrices (\ie $\mathcal{M}_r$ and $\mathcal{M'}_r$)  and to access to the $Leaf$ array. 
Another issue is that both these approaches have not been designed to process data streams. Ohrimenko et al.' solution  only focuses on inferences. Secure XGBoost supports generic decision tree models and is not designed for HT.  

%\notegio{Do these solution proposed by ohrimenko etal use oblivious execution in the enclave? How do they perform compared to ours?} 
%\notegio{Could they also handle data streams? if not why not?}
%\notewang{Here, our logic should be: Ohrimenko only provide DT inference; XGBoost provide static BoostTree(DT) training and inference based on Ohrimenko, but with heavy usage of primitives. I would modify below context as follows:} 
%\notewang{The baseline  in our work is implemented based on the design of \cite{law2020secure}.}

\subsection{Privacy-preserving Data Stream Mining}
In the literature, several works have focused on protecting data stream privacy  \cite{chamikara2019efficient,li2007hiding,zhou2009continuous,kellaris2014differentially}. However, they mainly focus on protecting the data distribution by adding noise. 
In more detail, these works leverage anonymization and data perturbation techniques to perturb the data and thus defend against attacks exploring the relationships across many features in data stream. 

Few works have considered protecting the training process and the generated model in data stream classification. 
For instance, the solution proposed in \cite{xu2008privacy} works on multiple stream sources to build a Na\"ive Bayesian model. 
They minimize the privacy leakage that could be incurred in the data exchange among data owners and do not consider the model privacy.
\cite{wang2019privstream} provides privacy protection for CNN inference with data stream but similarly the privacy of model and training process is not their focus.
While these two works focus on data streams, neither of these two schemes focus on data stream classification using HT. 
Moreover, the main drawback of both approaches is that frequently adding noise reduces the model accuracy which may require frequent reconstructions of the model.
% \sout{For instance, the solutions proposed in \cite{xu2008privacy} and \cite{wang2019privstream} introduce perturbations into the data stream or anonymize the model parameters. 
% The main drawback of both approaches is that it requires frequent reconstruction of the model in order to address the decreasing accuracy of the model. }
Another issue is that an attacker could infer sensitive information from the data stream, such as the user's identity, the locations a commuter visits and the type of illness a patient suffers from, by deploying various inference-based attacks~\cite{chamikara2019efficient,aggarwal2005k,kellaris2014differentially}.

%\notegio {I am not sure what this means: "due to the user-controlled privacy level"}
%\notewang{user-controlled privacy level means: user set a privacy budget to control the privacy leakage. It is acceptable if the leak is less than the budget. This privacy protection is not as strong as cryptography. So the attacker can leverage this property to infer the sensitive information.}

% \subsection{Homomorphic Encryption}

% Homomorphic Encryption is one of the most promising choices in modern cryptography but comes with heavy computational overheads. 
% Prior encryption systems of ElGamal \cite{elgamal1985public} and Paillier \cite{paillier1999public} only support one operation such as addition or multiplication. 
% Gentry constructed the first fully homomorphic encryption scheme (FHE) \cite{gentry2009fully}, which evaluates an arbitrary number of the above two operations on encrypted data. 
% Since then, a set of practical HE schemes have been proposed. 
% Compared to FHE, somewhat homomorphic encryption (SHE) schemes \cite{brakerski2011fully,naehrig2011can} are more efficient, which allow only a limited number of operations on ciphertexts.
% To further improve the performance of HE computations, recent studys \cite{juvekar2018gazelle,boemer2020mp2ml} introduce the ciphertext packing techniques \cite{smart2014fully} in ML tasks to perform HE computations in a SIMD manner.

%% file: sections/conclusion.tex
\section{Conclusion and Future Work}
\label{sec: conclusion}
We presented \sys, a practical, the first privacy-preserving data stream classification framework, which protects user's private information and the target model against access-pattern-based attacks.
\sys adopts novel matrix-based data-oblivious algorithms for the SGX enclave and uses x86 assembly oblivious primitives. 
\sys supports strong privacy guarantees while achieving acceptable performance overhead in privacy-preserving training and inference over data streams. 
As future work to improve \sys  performance, we will investigate two  potential solutions: (a) distribute the computation across multiple enclaves on different machines to perform matrix multiplications in parallel, and (b)  securely outsource the matrix multiplication to GPUs.
%\notewang{Do I need to provide more details for this future work? This would be: Deploy two GPU servers, outsource the share of the matrices (either $M_q$ or $M_t$) using additive secure sharing.}\notescui{No.}

%% file: sections/appendix.tex
% \section{Appendix}

\section{Security Analysis}
\label{sec:proof}
% 1.public information
% 2.trace of observations
%\sys provides provable security guarantees, ensuring that the memory access patterns do not reveal any information about sensitive data when running enclave code.  
%We adopt the definitions in \cite{ohrimenko2016oblivious,poddar2020visor} to formalize the security definitions. For our algorithms, we specify the \textit{public information} that is allowed to be disclosed to the adversary (such as some input sizes and the number of iterations).Then, we define a \textit{trace of observations} that the adversary sees. More specifically, it is a sequence of the addresses of memory references to the code and data, each recording an access type (read or write), an address, and some (encrypted) contents.

% 1.simulation paradigm
% 2.indistinguishable of trace
%We formulate our security properties using the “simulation paradigm” \cite{goldreich2004foundations}. We argue that given the inputs for an algorithm execution, there exists a \textit{simulator} based only on the public information. The simulator can produce a trace that is indistinguishable from the real trace, which is disclosed to the adversary who observes the access patterns during the training and inference phases. We define the indistinguishable property as no computationally bounded (\eg polynomial-time) adversary can distinguish between the simulated and real traces.It implies that the adversary cannot learn any private data if a simulator can make the same observations as produced by the adversary even without knowing the private data. 

In this section, we analyse how \sys protects the enclave access pattern along with detailed pseudocode. 
%\notegio{CHECK THIS PARAGRAPH AS IT DOES NOT PARSE: In short, we demonstrate that our algorithms are \textit{data-oblivious} by demonstrating that a sequence of memory accesses generated during the executions of our algorithms that depend only on public information and are independent of private data.} \notescui{to be discussed}

\begin{definition}[Data-oblivious]
As defined in \cite{poddar2020visor},  
we say that an algorithm is \textit{data-oblivious} if an adversary that observes its interaction with memory, disk or network during the executions learns only the public information.
\end{definition}

In the following, we prove both the HT training and HT inference in \sys is data-oblivious. 

% Given some input $x$, \ie a sequence of the addresses, along with the encrypted contents that are read (written) to the addresses, let $\mathtt{trace}_{\mathcal{A}}(x)$ be the trace of observations that the adversary yields during the execution of an algorithm $\mathcal{A}$. 
% To prove that $\mathcal{A}$ is data-oblivious, we show that there exists a simulator program given only the public information and regardless of the input $x$ that can produce a trace $\mathtt{\tau}_{\mathcal{A}}(x)$ indistinguishable from the trace $\mathtt{trace}_{\mathcal{A}}(x)$. 
% More precisely, we picture the following game: a program or the simulator that runs the algorithm $\mathcal{A}$. 
% A computationally bounded adversary takes public parameters as the inputs and observes the trace during the execution. 
% After terminating, the adversary's output consists of a trace of observations (say $\mathtt{\xi}_{\mathcal{A}}(x)$).
% She attempts to guess whether $\mathtt{\xi}_{\mathcal{A}}(x)$ was produced by the program or the simulator.
% $\mathcal{A}$ is data-oblivious when such adversary guess correctly with probability at most $1/2$ and a negligible advantage $\epsilon$.

% We now prove the security of our algorithms based on the above security assumptions in the following subsections.

\subsection{Oblivious HT Training}
Algorithm~\ref{alg3} provides the pseudocode of oblivious HT training.

% disable auto underline
\normalem
\begin{algorithm}[!h]
\footnotesize
\DontPrintSemicolon
\caption{Oblivious HT Training}
\label{alg3}
\SetKwFunction{GetPath}{GetPath}
\SetKwFunction{GenerateMasks}{GenerateMasks}
\SetKwFunction{GenQueryMatrix}{GenQueryMatrix}
\SetKwFunction{MatMul}{MatMul}
\SetKwFunction{RecordStat}{RecordStat}
\SetKwFunction{UpdateStat}{UpdateStat}
\SetKwFunction{MajorityClass}{MajorityClass}
\SetKwFunction{SplitCheck}{SplitCheck}
\SetKwFunction{CreateChildren}{CreateChildren}
% \SetKwInOut{Input}{Input}\SetKwInOut{Output}{Output}

\KwIn{$N$ encrypted data samples $Enc.D$, $S$, $V$, $m_{i}$, $M$}
%$d$: number of features
%\nonl \emph{\% Initialize $l$, $Leaf$, and $\mathcal{M}_t$ with dummy objects (\ie paths or nodes)}\;
Initialize the model matrix $\mathcal{M}_t$ and leaves array $Leaf$ with $P$ dummy objects. 
Initialize the bit string $isDummy$. 
Initialize a list $node$, where $node[p]=(fIdx, \tau_p)$ for $p \in [1, P]$. $node[p].fIdx$ stores the indices of the features that have not assigned on $\mathcal{M}_t[;p]$; and $node[p].\tau_p$ stores the number of feature values assigned to $\mathcal{M}_t[;p]$\; \label{alg3:line1}

\nonl \emph{\% Generate data sample matrix $\mathcal{M}_d$}\;
Decrypt $Enc.D$ and pack them into a $N{\times}M$ matrix $\mathcal{M}_d$\; \label{alg3:line3}

\nonl \emph{\% Generate query matrix $\mathcal{M}_q$}\;
% $vec$=[]\; \label{alg3:line4}
\ForEach{$p \in [1,P]$} 
{
    $T_{p} = \mathtt{oaccess}$($node[p].fIdx, S, V$), where $V=(V_{s_1}, ..., V_{s_d})$\;\label{alg3:line7}
    $\mathcal{M}_q^{p}$ = \GenerateMasks{$\mathcal{M}_t[; p]$, $isDummy$, $T_{p}$}\;\label{alg3:line8}
}
$\mathcal{M}_q = \mathcal{M}_q^{1} || \cdots || \mathcal{M}_q^{P}$ \; \label{alg3:line10}
%\GenQueryMatrix{$vec$}
\nonl \emph{\% Update $node$, $Leaf$, and $\mathcal{M}_t$ using $\mathcal{M}_r$}\;

$\mathcal{M}_r$ = \MatMul{$\mathcal{M}_d$, $\mathcal{M}_q$}\;\label{alg3:line11}
$output$=[]\;\label{alg3:line12}
% $output$ = \RecordStat{$\mathcal{M}_r$}\;\label{alg3:line13}
\ForEach{$p \in [1,P]$}
{
    $output$ = \RecordStat{$\mathcal{M}_r$,$node[p].\tau_p$}\;\label{alg3:line13}
    \UpdateStat{$isDummy$, $output$, $Leaf$}\;\label{alg3:line15}
    \nonl \emph{\% Check for a split}\;
    splitIdx = \SplitCheck{$isDummy$, $S$, $Leaf$}\;\label{alg3:line16}
    \nonl \emph{\% Generate new leaf nodes, update $l$, $Leaf$, and $\mathcal{M}_t$}\;
    $isSplit$ = (splitIdx == (-1))\;\label{alg3:line17}
    \CreateChildren{$isSplit, node, \mathcal{M}_t, Leaf$}\;\label{alg3:line18}
    %$\mathtt{oassign(isSplit{\wedge}n_i.type, i, \mathcal{M}_t)}$\;\label{alg3:line19}
    %$\mathtt{oassign(isSplit{\wedge}n_i.type, i, Leaf)}$\;\label{alg3:line20}
}
\end{algorithm}

\begin{theorem}
The oblivious HT training of \sys (Algorithm~\ref{alg3}) is data-oblivious with public parameters: $N$, $P$, $d$ and $M$.
\end{theorem}

\begin{proof}
%Given the $N$ encrypted data samples, the simulator executes Algorithm~\ref{alg3} and outputs the trace produced by the algorithm.
Here we analyse what the adversary can learn from each operation in Algorithm~\ref{alg3}. 

The memory access occurred due to the initialization (line~\ref{alg3:line1}) and $\mathcal{M}_d$ generation (line~\ref{alg3:line3}) is independent of the data, from which the adversary could only learn the size information $P$, $N$ and $M$, which are public.
% Line~\ref{alg3:line4} has fixed access patterns.

The loop from line~3 to line~\ref{alg3:line8} aims to traverse $\mathcal{M}_t$ and $node$ and generate the query matrix for each column of $\mathcal{M}_t$ . 
This loop always runs $P$ times, which means all the columns and elements of $\mathcal{M}_t$ and $node$ respectively are always accessed for each round of training, resulting the same access pattern no matter what the input is. 
Recall that the query matrix is generated based on the feature values assigned and those unassigned to the column. 
Within the loop, the enclave first fetches the values of unassigned features indexed by $node[p].fIdx$ from $S$ and $V$ using $\mathtt{oaccess}$ and stores them into $T_p$ (line~\ref{alg3:line7}).
Although $node[p].fIdx$ is different for different columns, the access patterns over $S$ and $V$ occurred by $\mathtt{oaccess}$ are oblivious and are independent of $node[p].fIdx$. 
The function $\mathtt{GenerateMasks}$ in line~\ref{alg3:line8} generates the query matrix based on the values in $T_p$ and $\mathcal{M}_t[; p]$.
$\mathcal{M}_t[; p]$ is obtained with $\mathtt{oaccess}$, which is also oblivious. 
Here the enclave generates query matrix in the same way for real and dummy columns. The difference is that the enclave assigns null to the masks for dummy columns, but the values in $T_p$ for real columns, however there is no way for the adversary to learn that. 
After the loop, the query matrix $\mathcal{M}_q$ of the whole tree is generated by combining the matrix of each path together.

Once $\mathcal{M}_d$ and $\mathcal{M}_q$ are ready, the next step is to perform the matrix multiplication, which is inherently oblivious, and obliviously access the result matrix $\mathcal{M}_r$ with oblivious primitives. 

The second loop (line~\ref{alg3:line13}-line~\ref{alg3:line18}) is used to update the statistic information stored in $Leaf$ and update $\mathcal{M}_t$ and $Leaf$ if they are leaves that need to be converted. 
The function $\mathtt{RecordStat}$ in line~\ref{alg3:line13} checks the elements in each column of $\mathcal{M}_r$ with $node[p].\tau_p$ and records the counts into a vector $output$. 
This process is performed obliviously with $\mathtt{oequal}$ and $\mathtt{oselect}$, resulting the access pattern over $\mathcal{M}_r$ and $output$ independent of any value.
In line~\ref{alg3:line15}, the enclave uses $\mathtt{oassign}$ to update $Leaf$ based on $output$. 
Here no matter whether the array is real or dummy, the enclave processes it with $\mathtt{oassign}$. The difference is that dummy arrays are assigned with $0$, but real arrays are assigned with the values recorded in $output$.
What the adversary observes from this process is all the same.

% The function $\mathtt{MajorityClass}$ in line~\ref{alg3:line19} assign the majority label for the leaf node $n_l$ using $\mathtt{ogreater}$ and $\mathtt{oselect}$, independent of any value. 
In line~\ref{alg3:line16}, the enclave checks whether to split the $p$-th leaf based on the updated $Leaf$. 
Precisely, the enclave first calculates the IG for all unassigned features. The enclave next uses $\mathtt{ogreater}$, $\mathtt{oequal}$ and $\mathtt{oselect}$ to select the feature with the highest and second-highest IG, return a value $splitIdx$ that indicates if $p$-th leaf is split by comparing with Hoeffding Bound (using $\mathtt{oselect}$). 
Its access patterns are thus independent of $S$. 

%Line~\ref{alg3:line17} has fixed access patterns. 
If $node[p]$ is real and its IG values satisfy the Hoeffding Bound, 
line~\ref{alg3:line18} converts the $p$-th leaf into internal nodes by updating $node$, $\mathcal{M}_t$ and $Leaf$ accordingly. 
The main idea is to convert $node[p]$, $\mathcal{M}_t[;p]$, and $Leaf[p]$ into dummies by resetting $isDummy$. 
Moreover, assume the best feature selected for converting the $p$-th leaf has $m$ values, $m$ dummies in $node$, $\mathcal{M}_t$ and $Leaf$ are converted into real ones by setting their values based on the new leaves and paths with oblivious primitives. 
If either $node[p]$ is dummy or it is not ready to be converted, the enclave similarly performs dummy write operations on $node$, $\mathcal{M}_t$ and $Leaf$, which is indistinguishable from the operations performed for the former case due to the oblivious primitives.

Overall, from Algorithm~\ref{alg3} the adversary can only learn the public information $N$, $P$, $d$ and $M$. 
%Thus, the trace produced by the simulator is indistinguishable from the trace produced by a real run of the Algorithm~\ref{alg3}.
\end{proof}

\subsection{Oblivious HT Inference}
\label{appendix:inference}

In this section, we provide pseudocode along with proofs of security for the oblivious HT inference in Algorithm~\ref{alg4}.

\begin{algorithm}[!h]
\footnotesize
\DontPrintSemicolon
\caption{Oblivious HT Inference}
\label{alg4}
\SetKwFunction{MajorityLabel}{MajorityLabel}
\SetKwFunction{GetIndiceLabel}{GetIndiceLabel}
\SetKwFunction{GenTreeMatrix}{GenTreeMatrix}
\SetKwFunction{MatMul}{MatMul}
\SetKwFunction{RecordStat}{RecordStat}
\SetKwFunction{Predict}{Predict}

\KwIn{$N'$ encrypted data instances $Enc.D$, $m_{i}$, $M$, $d$, $\mathcal{M}_t$, $Leaf$}

Decrypt the unlabelled instances and pack them into a $N'{\times}(M-m_d)$ matrix $\mathcal{M}_i$\; \label{alg4:line1}
Initialize label array $A_{label}$ of size $P$ for storing labels\; \label{alg4:line2}

\nonl \emph{\% Store labels in an array}\;
\ForEach{$p \in [1, P]$}
{
    $A_{label}$ = \MajorityLabel{$Leaf[p]$}\;\label{alg4:line4}
}

\nonl \emph{\% Record counts for each instance using $\mathcal{M}_r'$}\;
$\mathcal{M}_r'$ = \MatMul{$\mathcal{M}_i$,$\mathcal{M}_t$}\;\label{alg4:line5}
$output$=[]\;\label{alg4:line6}
$output$ = \RecordStat{$\mathcal{M}_r$}\;\label{alg4:line7}
\nonl \emph{\% Compare values in $output$ and assign labels to instances}\;
$Result$=[]\; \label{alg4:line8}
$Result$ = \Predict{$output$,$A_{label}$}\; \label{alg4:line9}
\Return{$Result$}\;
\end{algorithm}

\begin{theorem}
The oblivious HT inference of \sys (Algorithm~\ref{alg4}) is data-oblivious, with public parameters $N'$, $P$ and $M$.
\end{theorem}

\begin{proof}
% The \sout{simulator} procedure receives a set of instances, then runs the algorithm, and outputs the produced trace.
The access patterns of line~\ref{alg4:line1} depend only on the number of instances $N'$ and $M-m_d$. 
Line~\ref{alg4:line2} depends on $P$. 

The loop in line3 and line~\ref{alg4:line4} is used to determine each leaf's label of the current tree, which executes $P$ times. 
Within function $\mathtt{MajorityLabel}$, the enclave only uses oblivious primitives, which does not leak any access patterns. 
Thus, the adversary could only learn $P$.

In line~\ref{alg4:line5}, the access patterns occurred by the matrix multiplication is inherently oblivious. 

%Line~\ref{alg4:line6} has the fixed access patterns. 
The function $\mathtt{RecordStat}$ in line~\ref{alg4:line7} checks the elements of each column in $\mathcal{M}_r'$ and records the counts into $output$. 
Similarly, the two operations are both performed with oblivious primitives, which do not leak access patterns. 
%Line~\ref{alg4:line8} has the fixed access patterns. 
The function $\mathtt{Predict}$ in line~\ref{alg4:line9} first compares the values in $output$ using $\mathtt{oequal}$. 
It then accesses the $A_{label}$ to get the target label and assigns it to the corresponding instances using $\mathtt{oassign}$. 
%\notewang{For dummy labels, it performs dummy write operations (by simply setting the input condition to false).}
In this process, the adversary could only learn $N'$ and $P$.

%Overall, the trace produced by the algorithm can be simulated only using $N'$, $P$ and $M$.
\end{proof}

\section{Performance of \sys}
\label{sub:HT_performance_appendix}

\begin{figure*}[htbp]
  \centering
    \begin{subfigure}{0.49\textwidth}
      \centering   
      \includegraphics[width=0.85\linewidth]{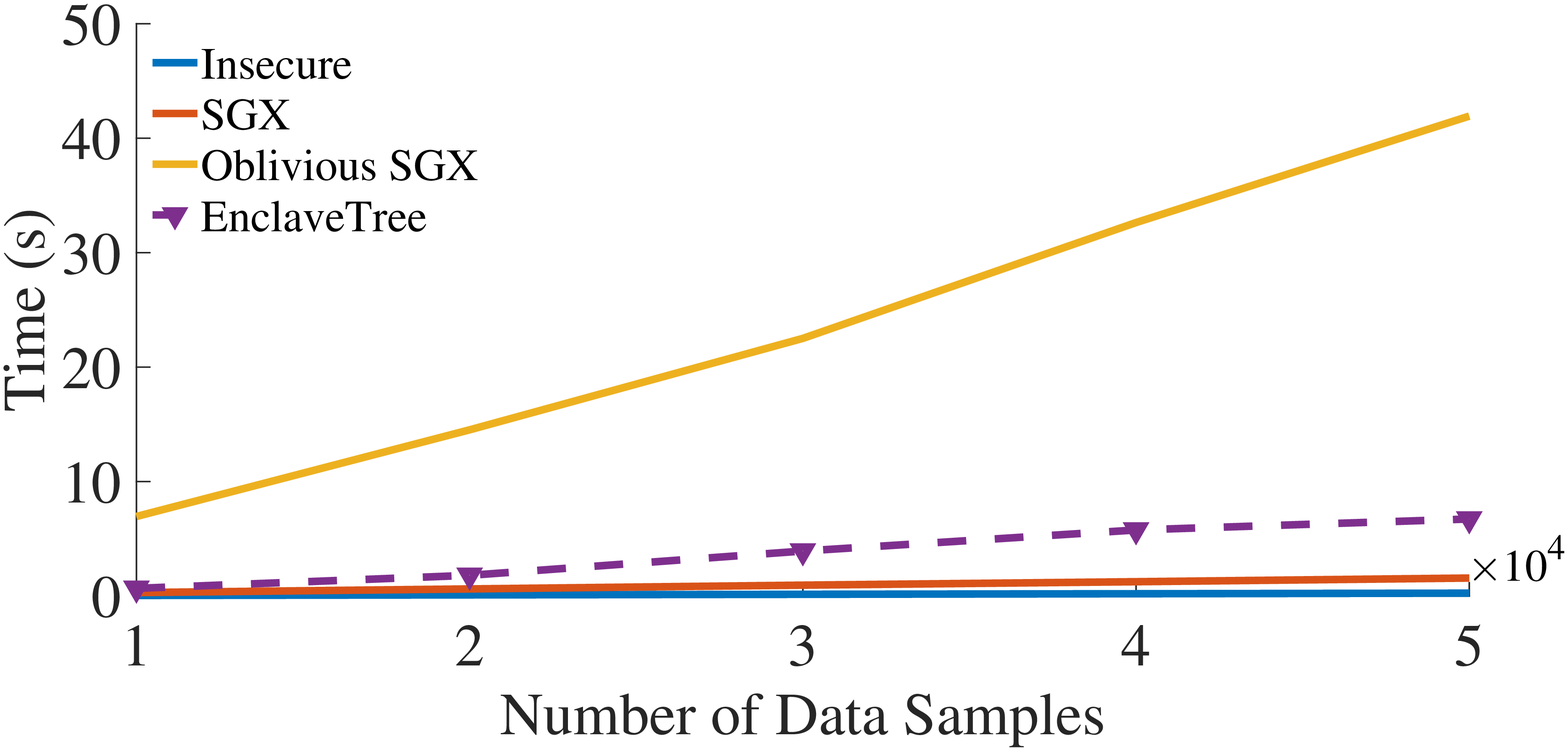}
        \caption{Training}
        \label{fig12:sub1}
    \end{subfigure}   %      \hfill  % 
    \begin{subfigure}{0.49\textwidth}
      \centering   
      \includegraphics[width=0.85\linewidth]{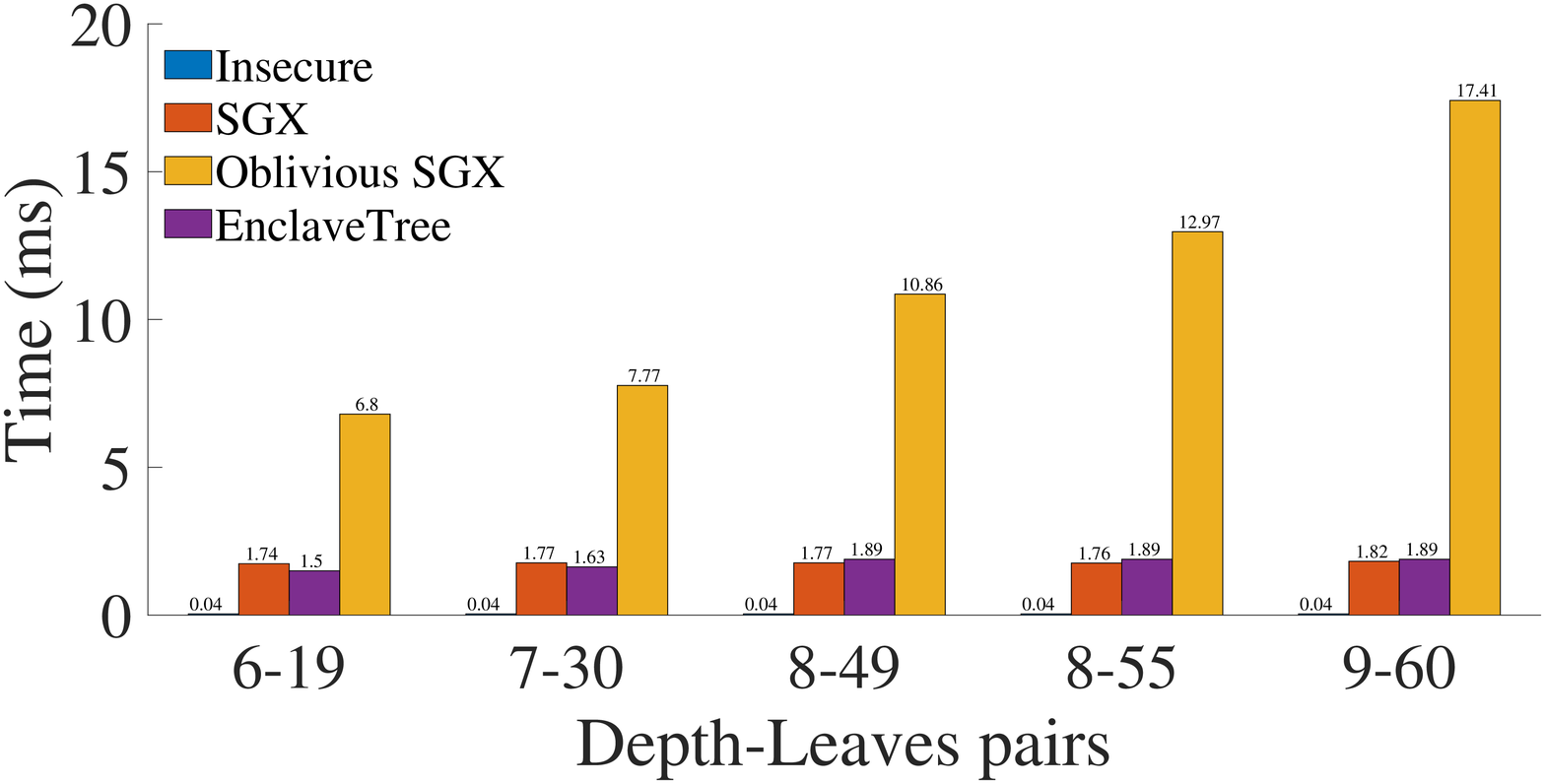}
        \caption{Inference}
        \label{fig12:sub2}
    \end{subfigure}
\caption{
\label{fig12}
The Comparison of HT Training and Inference with 15 features
}
\end{figure*}

\begin{figure*}[htbp]
  \centering
    \begin{subfigure}{0.49\textwidth}
      \centering   
      \includegraphics[width=0.85\linewidth]{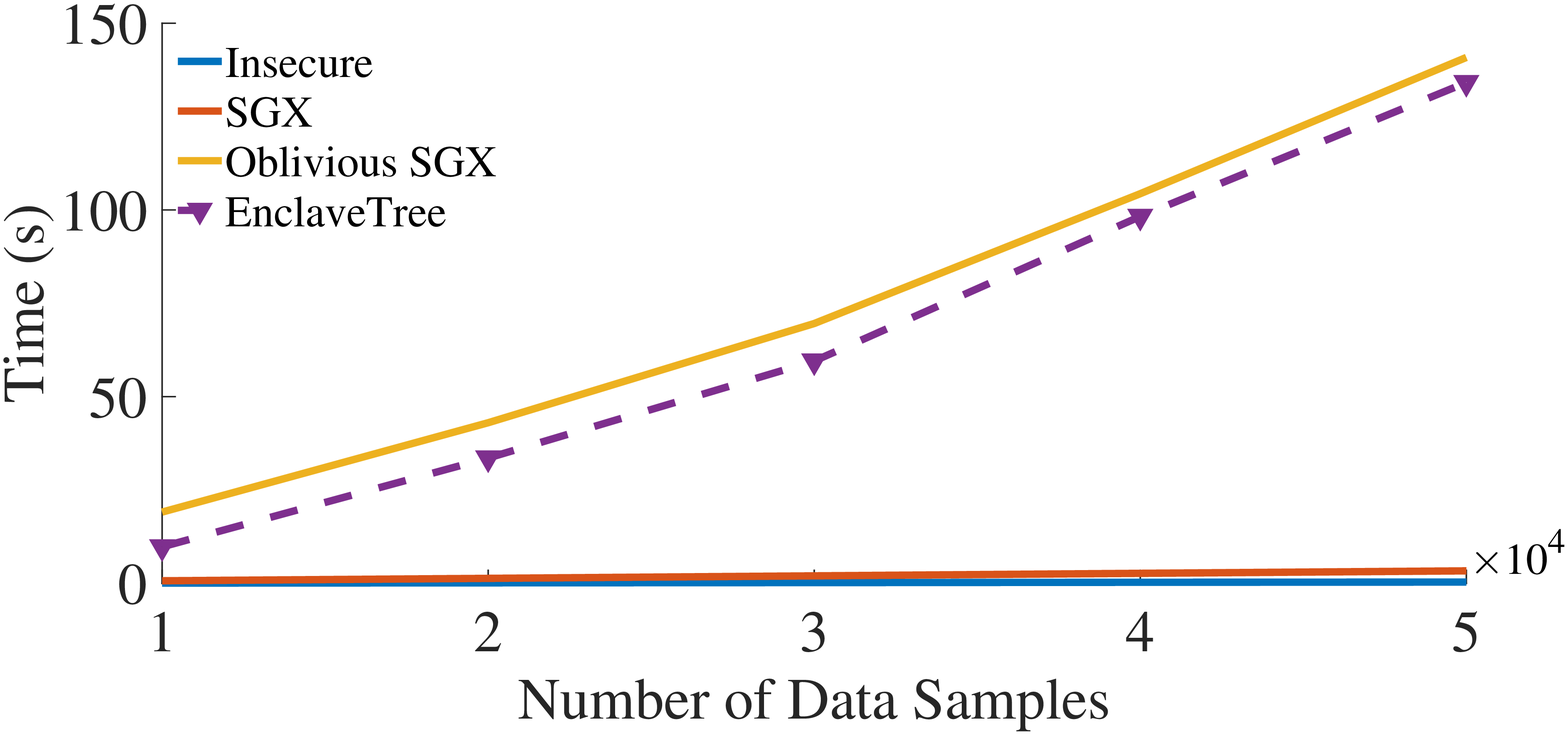}
        \caption{Training}
        \label{fig13:sub1}
    \end{subfigure}   %      \hfill  % 
    \begin{subfigure}{0.49\textwidth}
      \centering   
      \includegraphics[width=0.85\linewidth]{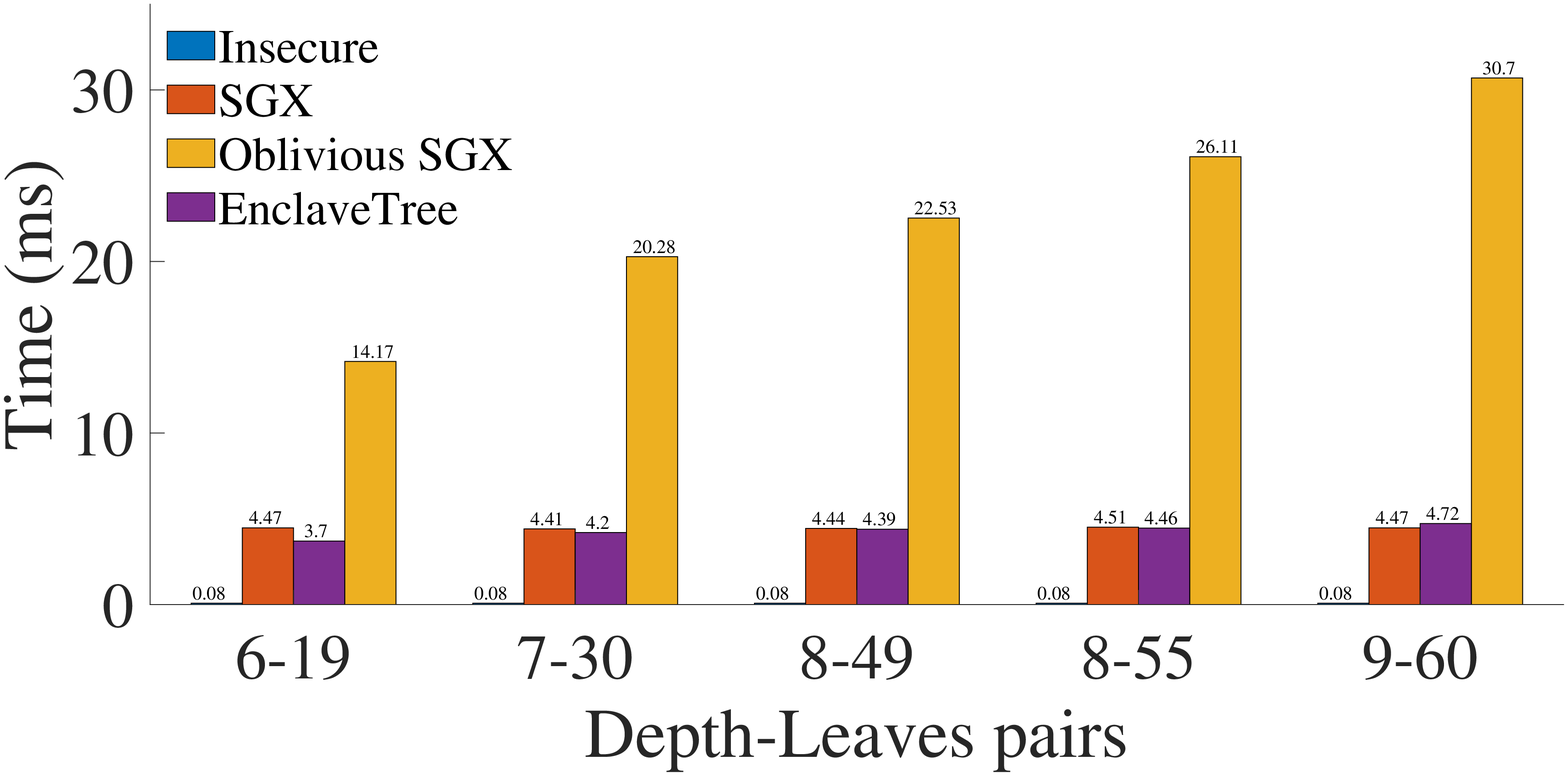}
        \caption{Inference}
        \label{fig13:sub2}
    \end{subfigure}
\caption{
\label{fig13}
The Comparison of HT Training and Inference with 63 features
}
\end{figure*}

\begin{figure*}[htbp]
  \centering
    \begin{subfigure}{0.49\textwidth}
      \centering   
      \includegraphics[width=0.85\linewidth]{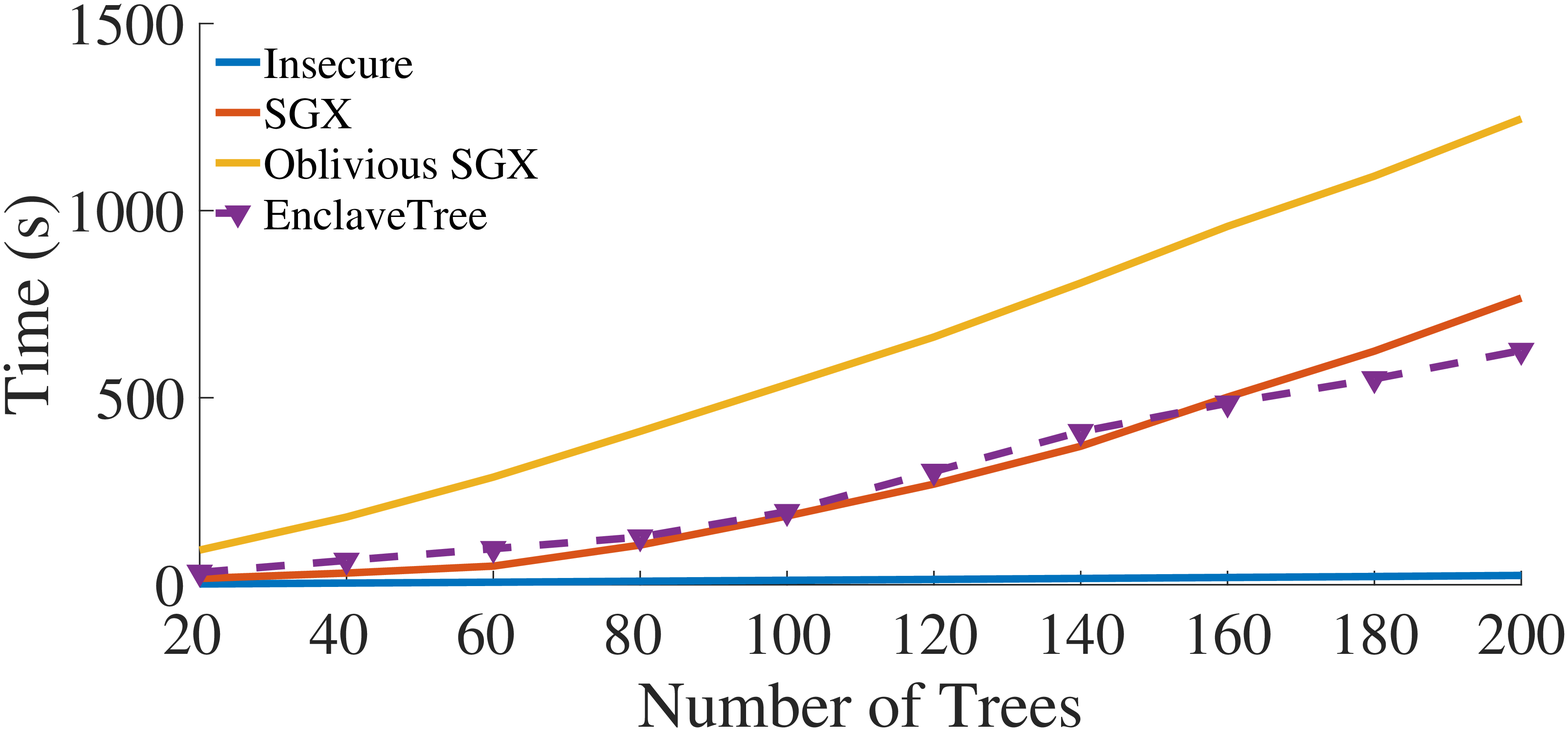}
        \caption{Training Runtime}
        \label{fig14:sub1}
    \end{subfigure}   %      \hfill  % 
    \begin{subfigure}{0.49\textwidth}
      \centering   
      \includegraphics[width=0.85\linewidth]{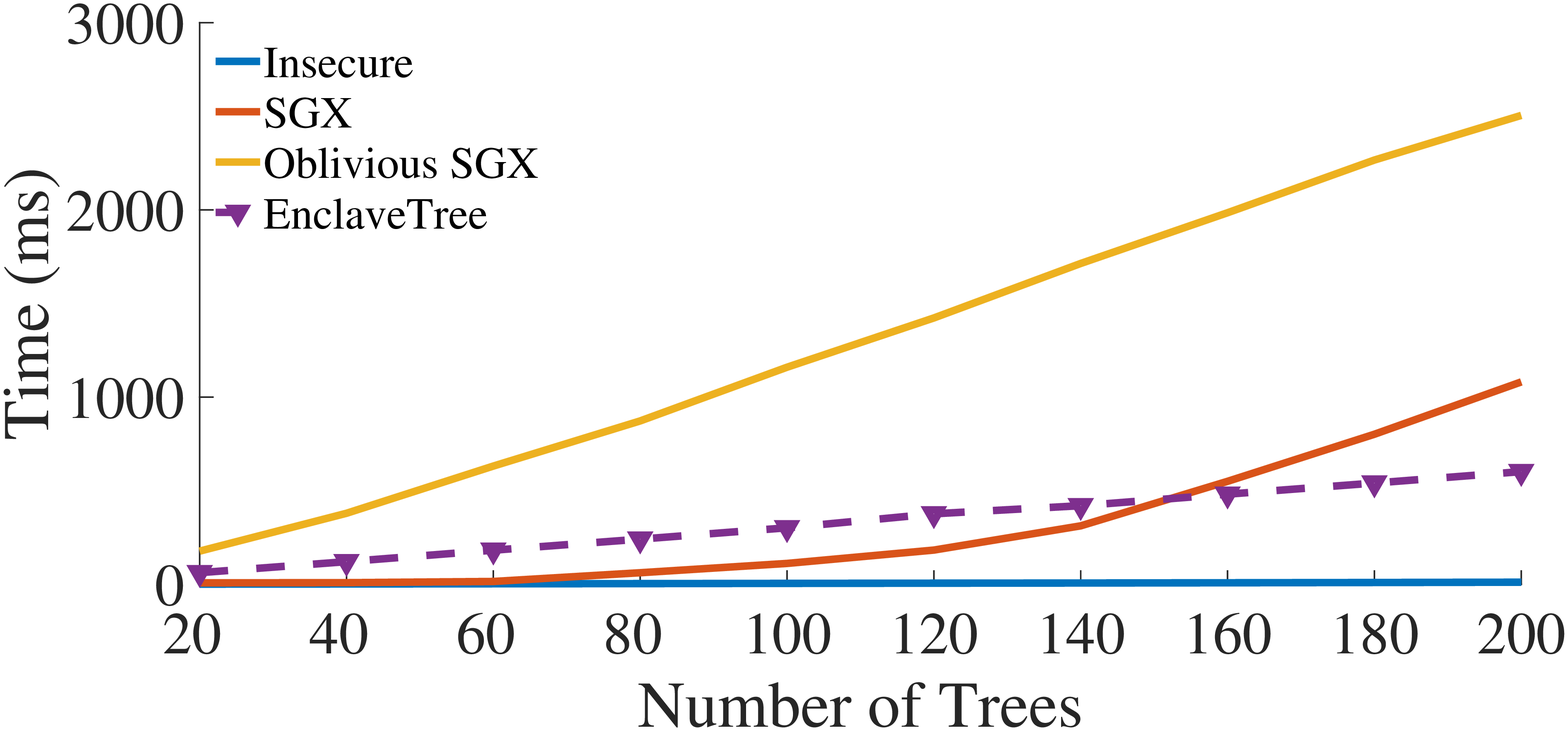}
        \caption{Inference Runtime}
        \label{fig14:sub2}
    \end{subfigure}   %      \hfill  % 
\caption{
\label{fig:RF_perf_result}
RF training and inference with the increase of trees
}
\end{figure*}

% \begin{figure*}[htbp]
%   \centering
%     \begin{subfigure}{0.49\textwidth}
%       \centering   
%       \includegraphics[width=0.8\linewidth]{graph/2.png}
%         \caption{Heap Memory Usage}
%         \label{fig13:sub2}
%     \end{subfigure}
%     \begin{subfigure}{0.49\textwidth}
%       \centering   
%       \includegraphics[width=0.8\linewidth]{graph/eva_epc_memory.eps}
%         \caption{Heap Memory Usage}
%         \label{fig13:sub2}
%     \end{subfigure}
% \caption{
% \label{fig13}
% RF Inference with the increase of trees
% }
% \end{figure*}

\subsection{More Results for HT Training and Inference}
Here we show the performance of HT training and inference with 15 features in Fig.~\ref{fig12} and 63 features in Fig.~\ref{fig13}.
It is indicated that HT training performs better with less number of features, which is close to $\mathtt{SGX}$ when there are 15 features. However, when the number of features increases to 63, the runtime of HT training is close to $\mathtt{Oblivious~SGX}$. 
The fact is that \sys is efficient to process the data streams in most scenarios as they generally involve about a dozen of features.  
Regarding the inference, as shown in Fig.~\ref{fig13}, our solution always outperforms $\mathtt{Oblivious~SGX}$ baseline by several orders of magnitude.

\subsection{RF Training and Inference}
\label{subsec:RF}
One concern of training data streams with HT is that the underlying data distribution of the stream might change over time, which leads to the accuracy degradation of the model, known as \textit{concept drift} \cite{gama2014survey}. Ensemble models such as \textit{Random Forest} (RF) with \textit{adaptive mechanisms} \cite{gomes2017adaptive} is a promising way to cope with the problem of concept drifts.

RF consists of a set of trees, and each tree is trained over a $\sqrt{d}$ subset of $S$ features. 
\sys uses the HT training component to train each tree in the RF. 
The features used to train a tree is randomly selected from $S$. 
To make the selection oblivious, the enclave accesses $S$ using $\mathtt{oaccess}$. 
Assigning a label to a data instance with RF inference means classifying the instance with each tree and getting a set of labels. The final result is the label that is output by the majority of trees. 

\sys performs the RF inference in a way similar to the HT inference using  matrix multiplication. 
In particular, the data instances can be classified by multiple trees with one matrix multiplication by combining the matrices of the trees together. 
 
We also evaluated the performance of RF training and inference, and the results are shown in Fig.~\ref{fig:RF_perf_result}. 
Fig.~\ref{fig14:sub1} shows the runtime in seconds to perform the RF training with 31 features and $5{\times}10^4$ samples. For all the test cases, every tree of the RF is trained with 6 features. 
The results show that, with the increase of trees, \sys is much faster than $\mathtt{Oblivious~SGX}$, which is up to ${\thicksim}3.2{\times}$. 
We also see that the performance of \sys is close to $\mathtt{SGX}$. 
%This is due to the performance drop caused by excessive memory usage inside the enclave with the increase of trees, as $\mathtt{SGX}$ and $\mathtt{Oblivious~SGX}$ store all of the nodes of each tree during the training. 

With the same setting, we compare the inference performance of \sys with the other three baselines and the result is shown in Fig.~\ref{fig14:sub2}. 
We can see that \sys also performs better than $\mathtt{Oblivious~SGX}$ by roughly $3.8{\times}$. 
Compared with $\mathtt{SGX}$, \sys inference incurs more overhead when there are less than about 150 trees but is better when there are more than 150 trees. 
The reason is that the inference process requires EPC memory to store data, and it causes EPC paging when the EPC is exhausted. 
\sys simply performs matrix multiplication, and this operation involves much less memory access than $\mathtt{SGX}$, which means less EPC paging occurred than $\mathtt{SGX}$. 

% \section{Discussion on Dummies}

% To hide the number of paths (\ie nodes) in the HT, \sys inserts a certain number of dummies. 
% However, for efficiency, $P_{dummy}$ is expected to be as small as possible. 
% To balance the efficiency and security, we carefully set the $P_{dummy}$ for HT in our implementation. 
% In the case of binary HT, assume the depth of the current HT is $d$, there are at most $2^{d-1}$ leaves. 
% \sys thus ensures $P$ is at least $2^{d-1}$ by inserting new dummies ($P=P_{real}+P_{dummy}$).
% This can avoid the case that dummies are not enough during training.
% %(see the notations in Section \ref{subsec:tree rep} and Table \ref{tab:notation}).
% Note that, in \sys, $P_{dummy}$ is dynamic during training since we convert the internal nodes into dummies. 
% Finally, the adversary can only learn $P$ and public parameter $d$.
% For non-binary HT, we cannot confirm the maximum number of nodes for a depth $d$ due to the non-binary features. 
% A large enough $P_{dummy}$ can guarantee security but degrade the performance. 
% We thus suggest a relatively small threshold $T$ to check if $P_{dummy}<T$ after each node split.
% If so, \sys inserts new dummies.
% Given different features in the datasets, $T$ is different. 
% One can experimentally select a proper $T$. 
% Overall, given proper $T$, \sys can avoid the worst case that all of the dummies are used or dummies are not enough. 
% The adversary can only learn $P$ and $T$ rather than the exact $P_{dummy}$ and $P_{real}$.